\documentclass[11pt,a4paper]{article}

\usepackage[a4paper,margin=1.2in]{geometry}
\usepackage{amsmath,amssymb,amsthm,mathtools}
\usepackage[dvipsnames]{xcolor}
\usepackage[colorlinks=true,citecolor=Blue,urlcolor=Blue,linkcolor=Blue]{hyperref}
\usepackage{enumerate}
\usepackage{natbib}

  \usepackage{tikz}
\usetikzlibrary {shapes}
\usetikzlibrary {arrows}
\usetikzlibrary {positioning}

\usepackage[OT2,OT1]{fontenc}
 \newcommand\cyr{%
 \renewcommand\rmdefault{wncyr}%
 \renewcommand\sfdefault{wncyss}%
 \renewcommand\encodingdefault{OT2}%
 \normalfont
 \selectfont}
 \DeclareTextFontCommand{\textcyr}{\cyr} 

\title{The Markov Marginal Problem for Density Operators}
\author{Steffen Lauritzen\thanks{Corresponding author:email \texttt{lauritzen@math.ku.dk}}\\University of Copenhagen \and Piotr Zwiernik\\Universitat Pompeu Fabra}
\date{}

\newcommand{\im}{\operatorname{im}}

\newcommand{\graph}{\mathcal{G}}
\newcommand{\hilb}{\mathcal{H}}
\newcommand{\cliques}{\mathcal{C}}
\newcommand{\sep}{\mathcal{S}}
\newcommand{\gen}{\mathcal{A}}
\newcommand{\trace}{\mathrm{Tr}}

\newcommand{\R}{\mathbb{R}}
\newcommand{\C}{\mathbb{C}}
\newcommand{\cL}{\mathcal{L}}

\renewcommand{\>}{\rangle}
\newcommand{\<}{\langle}
\newcommand{\cR}{\mathcal{R}}

\newcommand{\wcomb}{\circledast}

\newcommand{\gse}{\mbox{\,$\perp\!\!\!\perp_\graph$\,}}
\newcommand{\ciph}{\mbox{\,$\perp\!\!\!\perp_Q$\,}}
\newcommand{\cd}{\,|\,}

\newtheorem{thm}{Theorem}[section]
\newtheorem{prop}[thm]{Proposition}
\newtheorem{lem}[thm]{Lemma}

\theoremstyle{definition}
\newtheorem{definition}[thm]{Definition}
\newtheorem{ex}[thm]{Example}
\newtheorem{rem}[thm]{Remark}

\begin{document}
\maketitle

\begin{abstract}
We study when local reduced density operators, viewed as quantum marginals, can
be assembled into a global quantum state with a prescribed Markov structure.
The starting point is a canonical logarithmic construction $T(\mathcal R)$, the
noncommutative analogue of the junction-tree formula for decomposable graphical
models. Unlike in the classical case, this formal construction may fail:
noncommutativity can prevent it from being a normalized state with the
prescribed marginals. We prove that this obstruction is captured exactly by a
trace condition. For two overlapping marginals, and for clique marginals on a
chordal graph, the condition $\trace(T(\mathcal R))=1$ is equivalent to the
existence of a quantum Markov completion. When it exists, the completion is
unique, equal to $T(\mathcal R)$, and selected by the maximum entropy principle.
In the two-clique case, we also give an equivalent conditional reconstruction
characterization: the two natural one-sided sandwich reconstructions agree if
and only if the trace condition holds. We introduce the global quantum
information $g{\rm I}(\graph)_\rho$ associated with a chordal graph $\graph$ and show
that it is a relative-entropy discrepancy from $\rho$ to the logarithmic
candidate, with a trace correction when the candidate is not normalized. We
also prove an intersection property for strictly positive quantum conditional
independence. Three-qubit Pauli examples illustrate how the quantum obstructions are
real: local consistency, feasibility, Markov feasibility, and maximum entropy
can all separate.
\end{abstract}

\noindent\textbf{Keywords:}
quantum conditional independence; chordal graph; density operator; quantum marginal problem;
maximum entropy; quantum Markov property; total correlation.


\section{Introduction}

A recurring problem in quantum information is to understand what can be inferred
from local views of a global state. A multipartite density operator determines
reduced density operators on its subsystems. Conversely, one may ask whether a
given family of reduced density operators is compatible with a global state,
and if so, whether there is a canonical way to choose such a state. This is
the quantum marginal problem. It is the noncommutative analogue of a familiar
problem in probability and statistics: reconstructing, or approximating, a
joint distribution from overlapping marginals.

In the classical case, the maximum entropy principle gives a natural answer.
Among all distributions with the prescribed marginals, it selects the one that
adds as little extra information as possible. For decomposable graphical models
this principle has a particularly simple form. If the prescribed marginals are
the clique marginals of a chordal graph and agree on overlaps, then the unique
maximum entropy completion always exists and is given by the junction-tree formula. Moreover, it is Markov
with respect to the graph; see, for example, \citet{lauritzen:26}. This is one
of the reasons chordal graphs play a central role in graphical models. It also
parallels the positive definite completion theorem for chordal graphs
\citep{grone:etal:84}. More broadly, marginal problems and their dual
formulations have a long history in probability and optimization; see, for example, 
\cite{kellerer:64,kellerer:64zfw,kellerer_1984} and \cite{vorobev:62,vorobev:63,vorobev:67}.

The quantum marginal problem has been studied from several complementary
viewpoints. One prominent direction concerns spectral compatibility: which
spectra of reduced density operators can arise from a common global state? This
line of work is closely related to representation theory;
see, for instance, \citet{klyachko2004quantum},
\citet{christandl_mitchison_2006}, and the overview in
\citet{walter_2014}. Explicit compatibility criteria are known in several
special cases, including constraints for one-qubit marginals of pure
multiqubit states \citep{higuchi_sudbery_szulc_2003} and low-dimensional
bipartite settings \citep{bravyi_2004}. From a computational point of view,
the general consistency problem for local density matrices is hard; see
\citet{liu_2006} for a precise statement. More recent work gives other forms
of compatibility criteria, including a countable family of inequalities that
is necessary and sufficient for compatibility \citep{fraser_2022}, and
entropy-based criteria for marginal compatibility motivated by many-body
systems \citep{kim_2021,kim_2021_simplification}.

The present paper studies a different but related question. Rather than
characterizing all compatible marginal families, we ask when locally consistent
marginals admit a completion with a prescribed quantum Markov structure. The
quantum case is more delicate than the classical decomposable case: reduced
density operators need not commute, and consistency on overlaps is no longer
enough for the junction-tree reconstruction to work. Even when a global
completion exists, the formal logarithmic expression suggested by the
classical formula may fail to be normalized, may fail to have the prescribed
marginals, or may fail to satisfy the expected Markov properties. Our main
results identify the additional noncommutative obstruction: the logarithmic
junction-tree candidate is a valid Markov completion precisely under a
trace-one condition.

The central object is a canonical logarithmic operator $T(\mathcal R)$
associated with a family $\mathcal R=\{\rho_C:C\in\mathcal C\}$ of prescribed
reduced density operators. In the
two-clique case with prescribed marginals on $A\cup C$ and $B\cup C$, where $A, B, C$ are disjoint and $V=A\cup B\cup C$, it is
$$
T(\mathcal R)
=
\exp\{\log\rho_{A\cup C}+\log\rho_{B\cup C}-\log\rho_C\}.
$$
For clique marginals on a chordal graph, the analogous construction adds the
logarithms of the clique marginals and subtracts the logarithms of the
separator marginals with their multiplicities. This is the direct
noncommutative analogue of the classical junction-tree formula. The main
question is when this formal logarithmic reconstruction is a genuine density
operator with the prescribed marginals.

Our first result answers this question for two overlapping marginals. We prove
that $\trace(T(\mathcal R))\le 1$ and that the trace-one condition
$\trace(T(\mathcal R))=1$ is equivalent to the existence of a quantum
conditionally independent completion. When this happens, the completion is
unique and is equal to $T(\mathcal R)$. We then give an equivalent
conditional reconstruction form of the same criterion. In classical notation,
a Markov distribution satisfies
$$
p(a,b,c)=p(a\mid c)p(b,c)=p(b\mid c)p(a,c).
$$
For density operators, the two corresponding one-sided reconstructions need
not agree. We show that their agreement is equivalent to normality of a simple
operator
\begin{equation}\label{eq:K}
K=\rho_{A\cup C}^{1/2}\rho_C^{-1/2}\rho_{B\cup C}^{1/2},    
\end{equation}
and, in this case, $T(\mathcal R)=KK^*=K^*K$. This gives a concrete
multiplicative counterpart to the logarithmic trace criterion and links the
completion problem to the equality case in monotonicity of quantum relative
entropy.

We then extend the trace criterion to clique marginals on chordal graphs. For
a pairwise consistent family $\mathcal R$ of strictly positive clique
marginals, we prove that $\trace(T(\mathcal R))=1$ is equivalent to the
existence of a quantum Markov completion with respect to the graph. When it
exists, this completion is unique, equal to $T(\mathcal R)$, and also the
unique maximum entropy element among all completions with the prescribed clique
marginals. Thus the classical chordal graph construction has an
exact quantum analogue, but only under an additional trace-one condition that
captures a genuinely noncommutative obstruction.

A guiding quantity in the paper is the \emph{global quantum information}
$g{\rm I}(\graph)_\rho$ of a state $\rho$ relative to a chordal graph $\graph$.
It compares the entropy of $\rho$ with the entropy predicted by the chordal
Markov formula from its clique and separator reductions. In the two-clique
case this is the conditional mutual information ${\rm I}(A:B\mid C)_\rho$, while
for the empty graph it becomes the quantum analogue of multiinformation
\citep{studeny:05}. We show that $g{\rm I}(\graph)_\rho$ is a relative-entropy
discrepancy between $\rho$ and the logarithmic candidate determined by its
clique marginals, with a trace correction when the candidate is not normalized.
When the trace-one condition holds, this reduces to the relative entropy from
$\rho$ to its canonical graphical Markov reconstruction.

We also record a structural fact about quantum conditional independence that is
needed for graphical arguments. In the strictly positive finite-dimensional
setting considered here, quantum conditional independence satisfies the
intersection axiom. This property is automatic for classical conditional
independence under positivity, but it is not formally shown in the quantum setting; in
particular, \citet{leifer_poulin_2008} left open whether the entropic quantum
conditional independence relation has the full graphoid property. Our proof is
based on the equality case in monotonicity of quantum relative entropy under
partial trace. We include the argument in
Appendix~\ref{app:petz-partial-trace}. The appendix also derives the
corresponding Petz reconstruction formula in the partial-trace case, since this
formula is used repeatedly in the paper: in the proof of the intersection property, in the
two-clique reconstruction criterion, and in the comparison with one-sided
conditional reconstructions.

{Our results are also related to the literature on equality in strong
subadditivity and quantum conditional independence.} Equality in strong
subadditivity, or equivalently vanishing quantum conditional mutual information,
has several known characterizations. \citet{ruskai_2002} gave the logarithmic equality
condition, \cite{petz2003monotonicity} related equality to sufficiency and recovery for monotonicity of
relative entropy, and \citet{hayden_jozsa_petz_winter_2004} gave a structural
decomposition theorem for states saturating strong subadditivity. The
conditional-density and graphical-model point of view was developed by
\citet{leifer_poulin_2008}, who introduced quantum analogues of Markov
networks, factor graphs, and belief propagation, and emphasized that several
classical graphical-model equivalences break down in the quantum setting.
Related equivalent conditions, including sandwich formulae involving the
operator $K$ in \eqref{eq:K}, 
were studied by \citet{zhang:13}. Our focus is different: we use these
two-clique identities as local building blocks, but our main results concern
prescribed marginals on chordal graphs. We identify the trace condition under
which the logarithmic chordal construction is a valid density operator, is
Markov, and is the unique maximum entropy completion.

The paper is organized as follows. Section~\ref{sec:preliminaries} collects
notation and basic facts about density operators, entropy, relative entropy,
and quantum conditional independence. Section~\ref{sec:markov-completions}
develops the main theory: the two-clique trace criterion, its
conditional reconstruction form, the chordal trace criterion, the global
information identity, and the maximum entropy characterization. Section~\ref{sec:examples} presents examples based on Pauli expansions, illustrating
the noncommutative obstructions to the classical chordal formula. Appendix
\ref{app:petz-partial-trace} contains a self-contained proof of equality in
monotonicity and the associated sandwich formula for the partial trace.

\section{Preliminaries}\label{sec:preliminaries}

In this section we describe our notation and collect a few elementary facts that will be
used repeatedly throughout the paper. Since several of our later arguments rely
on concrete matrix manipulations, we keep the presentation explicit.

\subsection{Basic setup}

Let $V$ be a finite set and let $\hilb_v$, $v\in V$, be finite-dimensional
complex Hilbert spaces. For $A\subseteq V$, write
$$\hilb_A=\bigotimes_{v\in A}\hilb_v,$$
with the convention $\hilb_\varnothing=\C$. Here $\bigotimes$ denotes the tensor product.
Let $\cL(\hilb_A)$ denote the space
of linear operators on $\hilb_A$. We equip $\cL(\hilb_A)$ with the
Hilbert--Schmidt inner product
$\<M,N\>:=\trace(MN^*)$. Write $\mathbb S(\hilb_A)$ for the real vector space of self-adjoint operators
on $\hilb_A$, $\mathbb S^+(\hilb_A)$ for the cone of positive definite
operators, and
$$\mathbb S_1^+(\hilb_A):=
\{\rho\in\cL(\hilb_A):\rho=\rho^*,\ \rho\succ0,\ \trace(\rho)=1\}$$
for the set of strictly positive density operators. 
We shall use two natural operations on the positive cone.

\begin{definition}
For $M,N\in\mathbb S^+(\hilb)$ define $M\odot N,M\star N\in \mathbb S^+(\hilb)$ by
$$M\odot N:=\exp(\log M+\log N)\qquad\mbox{and}\quad\qquad M\star N:=N^{1/2}MN^{1/2}.$$
\end{definition}

The operation $\odot$ is commutative and associative and corresponds to
addition after applying the matrix logarithm. The operation $\star$ is neither commutative nor associative in general, but it
will be useful later for conditional reconstruction and behaves well under
partial trace. Note also that
$M\odot N=M\star N$
if and only if $M$ and $N$ commute in which case both expressions are also equal to $MN$. 

\subsection{Partial trace}

The quantum analogue of marginalization is the partial trace. Thus, whenever we speak of a marginal of a density operator, we mean the corresponding reduced density operator obtained by tracing out the complementary subsystem. We recall the basic facts we need; see, for example, Section~2.4.3 in \cite{nielsen:chuang:00} for more details. 

Let $A,B\subseteq V$ be disjoint finite sets. The \emph{partial trace over $A$} is the
unique linear map
$\trace_A:\cL(\hilb_{A\cup B})\to \cL(\hilb_B)$
such that
\begin{equation}\label{eq:partial-trace-def}
\trace\bigl((I_A\otimes M)\rho\bigr)
=
\trace\bigl(M\,\trace_A(\rho)\bigr)
\end{equation}
for all $\rho\in\cL(\hilb_{A\cup B})$ and all $M\in\cL(\hilb_B)$. Equivalently,
under the identification
$\cL(\hilb_{A\cup B})\cong \cL(\hilb_A)\otimes \cL(\hilb_B)$,
it is the linear map determined by
\begin{equation}\label{eq:partial-trace-simple}
\trace_A(X\otimes Y)=\trace(X)\,Y,
\qquad
X\in\cL(\hilb_A),\ Y\in\cL(\hilb_B).
\end{equation}

We note that the partial trace is linear, positive, and trace-preserving. 
In the rest of the paper we use the words \emph{marginal} and \emph{reduced density operator} interchangeably. We shall also use repeatedly that marginalization can be iterated.

\begin{lem}\label{lem:iterated-partial-trace}
Let $D\subseteq E\subseteq V$ and let $\rho\in\mathbb S_1^+(\hilb_V)$. Then
$$
\trace_{E\setminus D}\bigl(\trace_{V\setminus E}(\rho)\bigr)
=
\trace_{V\setminus D}(\rho) .
$$
\end{lem}

\begin{proof}
It suffices to test both sides against an arbitrary operator
$M\in\mathcal L(\hilb_D)$. By the defining property of the partial trace,
$$
\begin{aligned}
\trace\left\{
\trace_{E\setminus D}\bigl(\trace_{V\setminus E}(\rho)\bigr)M
\right\}
&=
\trace\left\{
(\trace_{V\setminus E}(\rho))(M\otimes I_{E\setminus D})
\right\}  =
\trace\left\{
\rho(M\otimes I_{E\setminus D}\otimes I_{V\setminus E})
\right\} \\
&=
\trace\left\{
\rho(M\otimes I_{V\setminus D})
\right\} =
\trace\left\{
(\trace_{V\setminus D}(\rho))M
\right\}.
\end{aligned}
$$
Since this holds for all $M\in\mathcal L(\hilb_D)$, the two operators are equal.
\end{proof}

The following pull-out property is one of the most useful identities in the
paper.

\begin{lem}[Pull-out property]\label{lem:pull-out} 
Let $A,B,C$ be pairwise disjoint, let $\rho\in\cL(\hilb_{A\cup B\cup C})$, and
let $M\in\cL(\hilb_{B\cup C})$. Then
$$\trace_A\bigl((I_A\otimes M)\rho\bigr)=M\,\trace_A(\rho),\qquad
\trace_A\bigl(\rho(I_A\otimes M)\bigr)=\trace_A(\rho)\,M.$$
\end{lem}

\begin{proof}
It is enough to verify the claim on simple tensors. If
$\rho=X\otimes Y$ with $X\in\cL(\hilb_A)$ and
$Y\in\cL(\hilb_{B\cup C})$, then
\begin{align*}
\trace_A\bigl((I_A\otimes M)(X\otimes Y)\bigr)
&=
\trace_A(X\otimes MY)
=
\trace(X)\,MY
=
M\,\trace(X)Y
=
M\,\trace_A(X\otimes Y),
\end{align*}
and the second identity is proved similarly.
\end{proof}

\subsection{Entropy and divergence}

We now recall the basic information-theoretic quantities used later. The
von Neumann entropy of a density operator $\rho\in \mathbb S_1^+(\hilb)$ is
$$
{\rm S}(\rho):=-\trace(\rho\log\rho).
$$
For our logarithmic constructions it is useful to extend this entropy to
positive definite operators whose trace is not necessarily one. The associated von Neumann divergence is
\begin{equation}\label{eq:divergence}
{\rm D}(X\|Y)
:=
\trace(X\log X)-\trace(X\log Y)-\trace(X)+\trace(Y).
\end{equation}
In particular, by Klein's inequality \citep[Theorem~3]{ruskai_2002}
\[
{\rm D}(X\|Y)\geq 0,
\qquad
{\rm D}(X\|Y)=0 \quad\Longleftrightarrow\quad X=Y.
\]
When $X$ and $Y$ have the same trace, and in particular when they are density
operators, the last two terms cancel. Thus for density operators
\[
{\rm D}(X\|Y)=\trace\{X(\log X-\log Y)\},
\]
which is the usual Umegaki quantum relative entropy; see, for example,
\citet[Section~11.3.1]{nielsen:chuang:00}.

\subsection{Chordal graphs}

We consider simple finite undirected graphs $\graph=(V,E)$. A subset
$A\subseteq V$ is \emph{complete} if every pair of distinct vertices in $A$ is
joined by an edge. A maximal complete subset is called a \emph{clique}, and we
write $\cliques$ for the set of cliques of $\graph$.

A subset $S\subseteq V$ is said to \emph{separate} $A\subseteq V$ from
$B\subseteq V$ in the graph $\graph$ if every path from a vertex in $A$ to a vertex in $B$ meets $S$.
We then write
$$A\gse B\cd S.$$
A \emph{decomposition} of $\graph$ is a triple $(A,B,S)$ such that
$V=A\cup B\cup S$,
the set $S$ is complete, and $A\gse B\cd S$.

We shall be particularly interested in \emph{chordal} graphs, that is, graphs
in which every induced cycle of length at least four has a chord, meaning an
edge joining two nonconsecutive vertices of the cycle. A basic fact is that the
cliques of a chordal graph can be arranged in a \emph{junction tree}: this is a
tree $\mathcal T$ with vertex set $\cliques$ such that whenever $C_1,C_2\in
\cliques$, every clique on the unique path between $C_1$ and $C_2$ contains
$C_1\cap C_2$. The \emph{separators} of $\graph$ are the intersections of pairs of
adjacent cliques in a junction tree,
$$\sep=\{C_1\cap C_2:\ C_1\sim C_2\text{ in }\mathcal T\},$$
and for $S\in\sep$, its multiplicity $\nu(S)$ is the number of times it appears
as such an intersection in any junction tree for $\cliques$. 
These notions are standard; see Fig.~\ref{fig:chordal} for an illustration and, for example,
\citet{lauritzen:26} for details. 
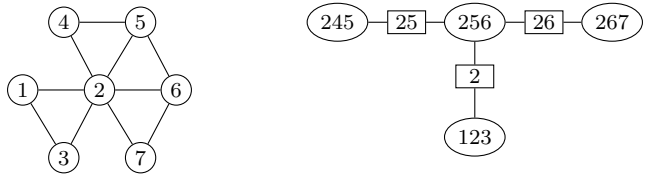
\begin{figure}[htb]
\begin{center}
\begin{tikzpicture}
   [node distance = 6mm and 6mm, minimum width = 4mm]
    \begin{scope}
      \tikzstyle{every node} = [shape = circle, 
      font = \scriptsize,
      minimum height = 4mm,
      inner sep = 0pt,
      draw = black, 
      fill = white, 
      anchor = center, 
      text centered] 

      \node(a) at (0,0) {1};
      \node(b) [right =  of a] {2};
	\node(c) [below right = 6mm and 2.5mm of a] {3};
	\node(d) [above right = 6mm and 2.5mm of a] {4};  
    \node(e) [ right =of d] {5};
	\node(f) [right = of b] {6};
    \node(g) [right = of  c] {7};
      \end{scope}

      \begin{scope}
      \tikzstyle{every node} = [shape = ellipse, 
      font = \scriptsize,
      minimum height = 5mm,
      minimum width = 8mm,
      inner sep = 0pt,
      draw = black, 
      fill = white, 
      anchor = center], 
      text centered] 

      \node (bde)[right  =  30mm of d] {245};
      \node (bef)[right = 10 mm of bde] {256};
      \node (bfg)[right =10mm of bef] {267};
      \node (abc)[below=10mm of bef] {123};
      \end{scope}
      
\begin{scope}
       \tikzstyle{every node} = [shape = rectangle, 
      font = \scriptsize,
      minimum height = 3mm,
      minimum width = 5mm,
      inner sep = 0pt,
      draw = black, 
      fill = white, 
      anchor = center], 
      text centered] 
      \node(be)[right = 2.5mm of  bde] {25};
      \node(bf)[right =2.5mm of bef] {26};
      \node (q) [below = 3mm of bef] {2};
\end{scope}

      \begin{scope}
          \draw (a)--(b);
          \draw (a)--(c);
          \draw (b)--(c);
          \draw (b)--(d);
          \draw (d)--(e);
          \draw(b)--(e);
          \draw (b)--(f);
          \draw (b)-- (g);
          \draw (e)--(f);
          \draw (f)-- (g);

          \draw(bde)--(be);
          \draw (be)--(bef);
          \draw (bef)--(bf);
          \draw (bf)--(bfg);
          \draw (q)--(bef);
          \draw (q)--(abc);
      \end{scope}

\end{tikzpicture}
\end{center}
\caption{A chordal graph and  associated junction tree with separators.}
\label{fig:chordal}
\end{figure}

\subsection{Quantum conditional independence}

Let $A,B,C\subseteq V$ be pairwise disjoint and let
$\rho\in\mathbb S_1^+(\hilb_{A\cup B\cup C})$. The \emph{quantum conditional
mutual information} is
$${\rm I}(A:B\cd C)_\rho:={\rm S}(\rho_{A\cup C})+{\rm S}(\rho_{B\cup C})-{\rm S}(\rho_C)-{\rm S}(\rho).$$
By strong subadditivity \citep{lieb_ruskai_1973}, this quantity is always
nonnegative.

\begin{definition}
We say that $A$ and $B$ are \emph{quantum conditionally independent given $C$}
with respect to $\rho$ if ${\rm I}(A:B\cd C)_\rho=0$. In this case we write
$$A\ciph B\cd C\,[\rho]$$ or just $A\ciph B\cd C$, when there is no ambiguity.
\end{definition}
Quantum conditional independence satisfies the so-called semi-graphoid axioms for an independence model \citep{pearl:88,studeny:05,lauritzen:26}, as also shown in \cite{leifer_poulin_2008}, so we have
\begin{enumerate}[(Q1)]
    \item $A\ciph B\cd C\implies B\ciph A\cd C$ (symmetry); 
    \item $A\ciph ( B\cup C) \cd D  \implies A\ciph B\cd D \text{ and } A\ciph C\cd D$ (decomposition);
   \item $A \ciph (B \cup D) \cd C\implies A \ciph B \cd (C\cup D)$ {(weak union)};
   \item  $A\ciph B \cd C \text{ and } A \ciph D \cd (B\cup C)\implies A \ciph (B\cup D) \cd C$ (contraction);
\end{enumerate}
In fact, the semi-graphoid properties follow directly from the strong subadditivity of the von Neumann entropy \citep[Lemma~5.1]{studeny:05}. 

Whether $\ciph$ satisfies intersection for positive density operators---i.e.\ is a graphoid independence model---was unclear \citep{leifer_poulin_2008}, but we can now confirm this positively:
\begin{prop}\label{prop:intersection}
For strictly positive density operators, quantum conditional independence satisfies intersection:
$$\text{\rm{(Q5)}} \quad A\ciph B \cd (C \cup D) \text{ and } A \ciph D\cd (B\cup C)\implies A \ciph (B\cup D) \cd C.$$ 
\end{prop}
    \begin{proof}See Appendix~\ref{sec:proof-inters}.
    \end{proof}

The intersection property ensures that different variants of Markov properties (pairwise, local, and global) are equivalent. The framework in \cite{leifer_poulin_2008} is based on the so-called local Markov property, whereas we are focusing on the global version, see Definition~\ref{def:quantum_markov} below.
\section{Markov and maximum entropy completions}\label{sec:markov-completions}

 We  introduce the logarithmic
construction $T(\mathcal R)$ for a pair of marginals and show that its trace detects the existence of a
Markov-type completion.

\subsection{The marginal problem for pairs}

We now apply this language to the first nontrivial marginal problem. The data are two reduced density operators, one on $A\cup C$ and one on $B\cup C$, where $A\cap B=\emptyset$. 

A necessary condition that the pair $\mathcal R=\{\rho_{A\cup C},\rho_{B\cup C}\}$ of reduced density operators can be obtained from a global state $\rho$ is that they are \emph{consistent}, meaning  that $\trace_{A}(\rho_{A\cup C})=\trace_{B}(\rho_{B\cup C})=:\rho_C$. 
Now let $M(\mathcal R)$ be the set of solutions to the marginal problem:
$$M(\mathcal R):=
\{\omega\in\mathbb S_1^+(\hilb_{A\cup B\cup C}):\omega_{A\cup C}=\rho_{A\cup C},\ \omega_{B\cup C}=\rho_{B\cup C}\}$$
and
\begin{equation}\label{eq:T-two-clique}
T(\mathcal R):=\exp\bigl(\log\rho_{A\cup C}+\log\rho_{B\cup C}-\log\rho_C\bigr)= \rho_{A\cup C}\odot\rho_{B\cup C}\odot\rho_{C}^{-1}.
\end{equation}

The operator $T(\mathcal R)$  should be viewed as the noncommutative analogue
of the classical conditional-independence extension. In the commuting case,
where $\rho_{A\cup C}$, $\rho_{B\cup C}$, and $\rho_C$ are simultaneously
diagonal, it reduces exactly to the
classical formula $p_{A\cup C}p_{B\cup C}p_C^{-1}$. In the noncommutative case,
however, the exponential expression need not behave like an ordinary product. In particular, it typically does not define a density operator. When it does, however, the expected Markov and maximum entropy properties coincide.

\begin{thm}[Two-clique trace criterion]\label{thm:main-trace2}
Let $\mathcal R=\{\rho_{A\cup C},\rho_{B\cup C}\}$ be a consistent pair of
strictly positive marginals. Then $\trace (T(\mathcal R))\le 1$. Moreover, the
following are equivalent:
\begin{enumerate}[(i)]
\item $\trace (T(\mathcal R))=1$;
\item $T(\mathcal R)\in M(\mathcal R)$;
\item there exists $\omega\in M(\mathcal R)$ such that
      $A\ciph B\cd C\,[\omega]$.
\end{enumerate}
When these conditions hold, $T(\mathcal R)$ is the unique Markov completion.
\end{thm}

The proof of Theorem~\ref{thm:main-trace2} uses the following divergence identity.
\begin{lem}\label{lem:two-clique-divergence-identity-general}
Let $\mathcal R=\{\sigma_{A\cup C},\sigma_{B\cup C}\}$ be a consistent pair of
density operators, and let $\omega\in\mathbb S_1^+(\hilb_{A\cup B\cup C})$.
Then
\begin{equation}\label{eq:T-divergence-general}
{\rm D}(\omega\|T(\mathcal R))+1-\trace(T(\mathcal R))
=
{\rm I}(A:B\cd C)_\omega+\Delta_{\mathcal R}(\omega),
\end{equation}
where
$$
\Delta_{\mathcal R}(\omega):=
{\rm D}(\omega_{A\cup C}\|\sigma_{A\cup C})+
{\rm D}(\omega_{B\cup C}\|\sigma_{B\cup C})-
{\rm D}(\omega_C\|\sigma_C).
$$
Moreover, $\Delta_{\mathcal R}(\omega)\ge 0$.
\end{lem}

\begin{proof}
Using \eqref{eq:T-two-clique} and the defining property of partial trace, we get
\begin{align*}
{\rm D}(\omega\|T(\mathcal R))
&=
\trace(\omega\log\omega)-\trace(\omega\log T(\mathcal R))-1+\trace(T(\mathcal R)) \\
&=
-{\rm S}(\omega)-\trace(\omega_{A\cup C}\log\sigma_{A\cup C})
-\trace(\omega_{B\cup C}\log\sigma_{B\cup C}) \\
&\qquad
+\trace(\omega_C\log\sigma_C)-1+\trace(T(\mathcal R)).
\end{align*}
Adding and subtracting
${\rm S}(\omega_{A\cup C})+{\rm S}(\omega_{B\cup C})-{\rm S}(\omega_C)$ yields
\eqref{eq:T-divergence-general}. The nonnegativity follows from monotonicity of relative entropy under partial
trace: for example,
$
{\rm D}(\omega_{B\cup C}\|\sigma_{B\cup C})
\ge {\rm D}(\omega_C\|\sigma_C).
$
This monotonicity is the partial-trace case of the data-processing inequality
for quantum relative entropy, following from the proof of strong
subadditivity in \cite{lieb_ruskai_1973}. For convenience, we recall it in Proposition~\ref{prop:app-monotonicity-partial-trace}.
\end{proof}

\begin{proof}[Proof of Theorem~\ref{thm:main-trace2}]
The trace inequality is now a consequence of Lieb's three-matrix inequality
\citep[Theorem~5]{ruskai_2002}:
\begin{align*}
\trace(T(\mathcal R))
&\le
\trace\!\left(
\int_0^\infty
\rho_{A\cup C}(tI+\rho_C)^{-1}\rho_{B\cup C}
(tI+\rho_C)^{-1}\,dt
\right).
\end{align*}
Since $(tI+\rho_C)^{-1}$ acts nontrivially only on $\hilb_C$,
Lemma~\ref{lem:pull-out} gives
\begin{align*}
&\trace\bigl(
\rho_{A\cup C}(tI+\rho_C)^{-1}\rho_{B\cup C}
(tI+\rho_C)^{-1}
\bigr) \; =
\;\trace\bigl(\rho_C^2(tI+\rho_C)^{-2}\bigr).
\end{align*}
Hence
$$
\trace(T(\mathcal R))
\le
\int_0^\infty \trace\bigl(\rho_C^2(tI+\rho_C)^{-2}\bigr)\,dt.
$$
Diagonalizing $\rho_C$ with eigenvalues $\lambda_i>0$, the right-hand side is
$$
\sum_i\int_0^\infty \frac{\lambda_i^2}{(t+\lambda_i)^2}\,dt
=
\sum_i\lambda_i
=
\trace(\rho_C)
=
1.
$$
This proves $\trace(T(\mathcal R))\le 1$.

For the second part of the theorem, set $\tau=T(\mathcal R)$. Suppose first that $\trace(T(\mathcal R))=1$.
Since $\tau\succ0$, $\tau$ is a density operator. Applying
\eqref{eq:T-divergence-general} with $\omega=\tau$ gives
$$
0
=
{\rm D}(\tau\|\tau)+1-\trace(\tau)
=
{\rm I}(A:B\cd C)_\tau+\Delta_{\mathcal R}(\tau).
$$
Both terms on the right-hand side  are nonnegative and so their sum is zero if and only if they are both zero. By the definition of $\Delta_{\mathcal R}$, this means
\begin{equation}\label{eq:delta-zero-two-clique}
{\rm D}(\tau_{A\cup C}\|\rho_{A\cup C})
+
{\rm D}(\tau_{B\cup C}\|\rho_{B\cup C})
-
{\rm D}(\tau_C\|\rho_C)
=
0.
\end{equation}
By monotonicity of relative entropy under partial trace (see \citet{lieb_ruskai_1973} or Proposition~\ref{prop:app-monotonicity-partial-trace}), we have
$$
{\rm D}(\tau_{A\cup C}\|\rho_{A\cup C})\geq {\rm D}(\tau_C\|\rho_C),
\qquad
{\rm D}(\tau_{B\cup C}\|\rho_{B\cup C})\geq {\rm D}(\tau_C\|\rho_C).
$$
Substituting these lower bounds into \eqref{eq:delta-zero-two-clique} gives
$0
\geq
{\rm D}(\tau_C\|\rho_C)$. Hence ${\rm D}(\tau_C\|\rho_C)=0$, and so $\tau_C=\rho_C$. Substituting this back into
\eqref{eq:delta-zero-two-clique} gives
$
{\rm D}(\tau_{A\cup C}\|\rho_{A\cup C})
+
{\rm D}(\tau_{B\cup C}\|\rho_{B\cup C})
=
0$. Thus
$
\tau_{A\cup C}=\rho_{A\cup C}$ and $\tau_{B\cup C}=\rho_{B\cup C}$.
Therefore $\tau\in M(\mathcal R)$. Since also ${\rm I}(A:B\cd C)_\tau=0$, we have
$A\ciph B\cd C\,[\tau]$. Hence $(\mathrm{i})$ implies both
$(\mathrm{ii})$ and $(\mathrm{iii})$. The implication $(\mathrm{ii})\Rightarrow(\mathrm{i})$ is immediate, since
every element of $M(\mathcal R)$ has trace one. Consequently, (i)$\Leftrightarrow$(ii). Suppose now that (iii) holds. Applying \eqref{eq:T-divergence-general} to this
$\omega$ gives
$$
0
=
{\rm I}(A:B\cd C)_\omega+\Delta_{\mathcal R}(\omega)
=
{\rm D}(\omega\|T(\mathcal R))+1-\trace(T(\mathcal R)).
$$
The two terms
on the right are nonnegative: ${\rm D}(\omega\|T(\mathcal R))\geq0$ by Klein's
inequality  and
$1-\trace(T(\mathcal R))\geq0$ by the trace bound proved above. Hence both
terms vanish. In particular,
${\rm D}(\omega\|T(\mathcal R))=0$,
and therefore $\omega=T(\mathcal R)$. Thus $T(\mathcal R)\in M(\mathcal R)$ and
$\trace(T(\mathcal R))=1$. This proves
$(\mathrm{iii})\Rightarrow(\mathrm{ii})$ and
$(\mathrm{iii})\Rightarrow(\mathrm{i})$. The same argument proves uniqueness. Any Markov completion
$\omega\in M(\mathcal R)$ must be equal to $T(\mathcal R)$.
\end{proof}

For a strictly positive state $\omega$, the condition
${\rm I}(A:B\cd C)_\omega=0$ is equivalent to the logarithmic identity
$$
\log\omega
=
\log\omega_{A\cup C}
+
\log\omega_{B\cup C}
-
\log\omega_C.
$$
This is one of the standard equality conditions for strong subadditivity and is
closely related to Petz's equality theorem for monotonicity of relative entropy. Theorem~\ref{thm:main-trace2} uses this equality theory in the reverse
direction. Starting only from a consistent pair
$\rho_{A\cup C},\rho_{B\cup C}$, it forms the logarithmic candidate
$T(\mathcal R)$ and shows that $\trace (T(\mathcal R))\le 1$. Moreover,
$\trace (T(\mathcal R))=1$ holds if and only if a Markov completion exists; in
that case the completion is necessarily $T(\mathcal R)$. 

The trace condition is genuinely stronger than ordinary feasibility. In
Example~\ref{ex:basic-qubit}, the marginals are locally consistent and strictly
feasible whenever $\varepsilon^2+\delta^2<1$, but
Lemma~\ref{lem:basic-qubit-no-markov} shows that they admit no Markov completion
when $\varepsilon\delta\ne0$. Equivalently, by
Lemma~\ref{lem:basic-qubit-trace-defect}, the logarithmic candidate has trace
strictly smaller than one in this case.

It should also be noted that if $\rho_{A\cup C}$, $\rho_{B\cup C}$ and $\rho_C$ all commute, the trace condition is always fulfilled as the operators are simultaneously diagonalizable. However, the trace condition may well be fulfilled without the commutation condition, see the remark after Theorem~5 in \cite{ruskai_2002}.

An alternative reconstruction of the joint density operator
combines appropriate conditional objects with the marginals. In the classical
Markov case, we have
$$
p(a,b,c)=p(a\mid c)p(b,c)=p(b\mid c)p(a,c).
$$
Thus the joint distribution is obtained by combining the
conditional factor with either of the two prescribed marginals. In the quantum
case these two reconstructions are no longer automatically the same.

For disjoint $A,B\subseteq V$ and
$\rho\in\mathbb S_1^+(\hilb_{A\cup B})$, define the \emph{direct conditional density
operator}
\begin{equation}\label{eq:direct-conditional}
\rho_{A|B}:=
\rho_B^{-1/2}\rho_{A\cup B}\rho_B^{-1/2}.
\end{equation}
This is one of several possible noncommutative conditional objects; see
\citet{leifer_poulin_2008}. Its role here is simple: it separates a joint
state into a conditional part and a marginal part. Namely,
$$
\rho_{A\cup B}=\rho_{A\cd B}\star\rho_B,$$
 where we recall that 
$M\star N:=N^{1/2}MN^{1/2}.$
The pull-out property gives $\trace_A(\rho_{A\cd B})=I_B$ and hence the conditional operator acts similarly to a Markov kernel. See also \cite{leifer2013towards} for this and an associated discussion.
We shall use the following reconstruction fact. 

\begin{lem}\label{lem:conditional-reconstruction}
Let $A,B,C$ be disjoint,
$\rho_{B\cup C}\in\mathbb S_1^+(\hilb_{B\cup C})$ and
$\sigma_{A\cup C}\in\mathbb S_1^+(\hilb_{A\cup C})$. Define
$$
\omega:=\rho_{B\cd C}\star\sigma_{A\cup C},
$$
where both factors are embedded to $\hilb_{A\cup B\cup C}$. Then
$\omega$ is a density operator with
$$
\omega_{A\cup C}=\sigma_{A\cup C}
\qquad\text{and}\qquad
\omega_{B\cd A\cup C}=\rho_{B\cd C},
$$
where the second identity uses the usual embedding of $\rho_{B\cd C}$ into
$\hilb_{A\cup B\cup C}$.
\end{lem}
Note that in this construction, $\rho_{B\cup C}$ and $\sigma_{A\cup C}$ do not have to agree on the $C$-margin.
\begin{proof}
Since $\sigma_{A\cup C}$ acts only on the subsystem not traced out, we get 
$$
\trace_B(X\star\sigma_{A\cup C})
=
\trace_B(X)\star\sigma_{A\cup C}
$$
for every positive operator $X$ on $\hilb_{A\cup B\cup C}$. Hence
$$
\omega_{A\cup C}
=
\trace_B(\rho_{B\cd C}\star\sigma_{A\cup C})
=
\trace_B(\rho_{B\cd C})\star\sigma_{A\cup C}
=
(I_A\otimes I_C)\star\sigma_{A\cup C}
=
\sigma_{A\cup C}.
$$
Thus $\omega$ has trace one, and it is positive by construction. Finally,
$$
\omega_{B\cd A\cup C}
=
\sigma_{A\cup C}^{-1/2}\omega\,\sigma_{A\cup C}^{-1/2}
=
\rho_{B\cd C},
$$
again with the usual embedding.
\end{proof}
For a consistent pair
$\mathcal R=\{\rho_{A\cup C},\rho_{B\cup C}\}$, the lemma gives two natural
one-sided reconstructions:
$$
\rho_{B\cd C}\star\rho_{A\cup C}
=
\rho_{A\cup C}^{1/2}\rho_C^{-1/2}
\rho_{B\cup C}
\rho_C^{-1/2}\rho_{A\cup C}^{1/2},
$$
and
$$
\rho_{A\cd C}\star\rho_{B\cup C}
=
\rho_{B\cup C}^{1/2}\rho_C^{-1/2}
\rho_{A\cup C}
\rho_C^{-1/2}\rho_{B\cup C}^{1/2}.
$$
Classically these two operators coincide. The next proposition says that, in the quantum case,
they coincide exactly when the  two-clique trace criterion is fulfilled.

The following result follows directly from \citep[Theorem~2.1]{zhang:13}. We include the proof for completeness.
\begin{prop}[Conditional reconstruction form of the two-clique criterion]
\label{prop:petz2}
Let
$$
K=\rho_{A\cup C}^{1/2}\rho_C^{-1/2}\rho_{B\cup C}^{1/2}.
$$
The equivalent conditions of Theorem~\ref{thm:main-trace2} hold if and only if
$K$ is normal. In that case
$$
T(\mathcal R)=KK^*=K^*K,
$$
or, equivalently,
$$
T(\mathcal R)
=
\rho_{B\cd C}\star\rho_{A\cup C}
=
\rho_{A\cd C}\star\rho_{B\cup C}.
$$
\end{prop}

\begin{proof}
Suppose first that the conditions of Theorem~\ref{thm:main-trace2} hold, and
set $\tau=T(\mathcal R)$. Then $\tau$ is the Markov completion of the prescribed
marginals, so ${\rm I}(A:B\cd C)_\tau=0$. By
Proposition~\ref{prop:app-petz-partial-trace}, the condition ${\rm I}(A:B\cd C)_\tau=0$ gives the
two Petz reconstruction identities
$$
\tau
=
\tau_{A\cup C}^{1/2}\tau_C^{-1/2}
\tau_{B\cup C}
\tau_C^{-1/2}\tau_{A\cup C}^{1/2}\qquad\mbox{and}\qquad
\tau
=
\tau_{B\cup C}^{1/2}\tau_C^{-1/2}
\tau_{A\cup C}
\tau_C^{-1/2}\tau_{B\cup C}^{1/2}.
$$
Since $\tau$ has marginals
$\rho_{A\cup C},\rho_{B\cup C},\rho_C$, these identities become
$
\tau=KK^*=K^*K$.
Thus $K$ is normal and $T(\mathcal R)=KK^*=K^*K$.

Conversely, suppose that $K$ is normal and set
$\omega=KK^*=K^*K$. Then $\omega\ge 0$. Using $\omega=KK^*$ gives
$$
\trace_B(\omega)
=
\rho_{A\cup C}^{1/2}\rho_C^{-1/2}
\trace_B(\rho_{B\cup C})
\rho_C^{-1/2}\rho_{A\cup C}^{1/2} =
\rho_{A\cup C}.
$$
Using $\omega=K^*K$ similarly gives
$\trace_A(\omega)=\rho_{B\cup C}$. Hence $\omega\in M(\mathcal R)$. Moreover, $\omega_C=\rho_C$, since it is the $C$-marginal of either
$\omega_{A\cup C}=\rho_{A\cup C}$ or
$\omega_{B\cup C}=\rho_{B\cup C}$. Therefore the identity $\omega=KK^*$ can be
rewritten as
$$
\omega
=
\omega_{A\cup C}^{1/2}\omega_C^{-1/2}
\omega_{B\cup C}
\omega_C^{-1/2}\omega_{A\cup C}^{1/2}.
$$
By Proposition~\ref{prop:app-petz-partial-trace}, this implies
${\rm I}(A:B\cd C)_\omega=0$. Thus $\omega$ is a Markov completion. By
Theorem~\ref{thm:main-trace2}, the Markov completion is unique and equal to
$T(\mathcal R)$. Hence
$
T(\mathcal R)=\omega=KK^*=K^*K,
$
as desired.
\end{proof}

\begin{definition}
When the equivalent conditions in Proposition~\ref{prop:petz2} hold, we say
that the pair $\rho_{A\cup C},\rho_{B\cup C}$ is \emph{Markov compatible over}
$C$, and we write
$$
\rho_{A\cup C}\wcomb\rho_{B\cup C}
$$
for the common operator $KK^*=K^*K$.
\end{definition}

\begin{rem}[Conditional interpretation]
Suppose that the pair $\rho_{A\cup C},\rho_{B\cup C}$ is Markov compatible over
$C$, and let
$
\omega=\rho_{A\cup C}\wcomb\rho_{B\cup C}.
$
Then the two one-sided conditional reconstructions agree:
$$
\omega
=
\rho_{B\cd C}\star\rho_{A\cup C}
=
\rho_{A\cd C}\star\rho_{B\cup C}.
$$
Equivalently, conditioning on the larger system adds no further information:
$$
\omega_{A\cd B\cup C}
=
\rho_{A\cd C}\otimes I_B,
\qquad
\omega_{B\cd A\cup C}
=
\rho_{B\cd C}\otimes I_A,
$$
with the usual embeddings.
\end{rem}

The proposition separates two roles of the classical factorization. The
logarithmic expression gives the variational candidate, while the sandwich
expression gives the conditional reconstruction candidate. In the commuting
case the two coincide automatically; in the noncommutative case their
coincidence is equivalent to  normality of $K$.

\subsection{Markov states on chordal graphs}

We now move from two overlapping reduced density operators to the family of
clique reductions associated with a chordal graph. Thus the local data are
density operators $\rho_C$ on the clique Hilbert spaces $\hilb_C$,
$C\in\cliques$, which are required to agree after partial trace on their
overlaps. The main point is that the two-clique picture survives in a
nontrivial way: the logarithmic construction still yields a canonical candidate
for a Markov completion, and its trace still detects whether such a completion
actually exists.

\begin{definition}\label{def:quantum_markov}
A density operator $\rho\in\mathbb S_1^+(\hilb_V)$ is \emph{quantum Markov} on
$\graph$ if for all disjoint subsets $A,B,C\subseteq V$ it holds that
$$A\gse B\cd C\implies A\ciph B\cd C\,[\rho].$$
\end{definition}
We note that this is the global variant of a range of alternative Markov properties that can be associated with an undirected graph.
For decompositions of chordal graphs, this Markov property has the usual
recursive form:
\begin{lem}\label{lem:decomp}
Let $\graph=(V,E)$ be a chordal graph and let $\rho$ be a density operator on
$\hilb_V$. If $(A,B,S)$ is a decomposition of $\graph$, then $\rho$ is quantum
Markov on $\graph$ if and only if $\rho_{A\cup S}$ and $\rho_{B\cup S}$ are
quantum Markov on the induced subgraphs $\graph_{A\cup S}$ and
$\graph_{B\cup S}$, and $A\ciph B\cd S\,[\rho]$.
\end{lem}

\begin{proof}
This is the standard recursive characterization of the global Markov property on
a chordal graph. The proof in \citet[Proposition~4.19]{lauritzen:26} uses only the semi-graphoid axioms (Q1)--(Q4) satisfied by the
conditional independence relation.
\end{proof}

The next proposition is the chordal analogue of strong subadditivity. It provides the
entropy characterization of the quantum Markov property that will be crucial
later.

\begin{prop}\label{prop:ssa-chordal}
Let $\graph$ be chordal, with clique set $\cliques$ and separator set $\sep$.
Then every $\rho\in\mathbb S_1^+(\hilb_V)$ satisfies
$${\rm S}(\rho)\le \sum_{C\in\cliques}{\rm S}(\rho_C)-\sum_{S\in\sep}\nu(S){\rm S}(\rho_S),$$
with equality if and only if $\rho$ is quantum Markov on $\graph$.
\end{prop}

\begin{proof}
The proof is by induction on the number $|\cliques|$ of cliques of $\graph$.
If $|\cliques|=1$, there is nothing to prove. If $|\cliques|=2$, the statement
is precisely strong subadditivity together with the definition of quantum
conditional independence.

Assume now that the result holds for all chordal graphs with at most $n$
cliques, and let $\graph$ be a chordal graph with $n+1$ cliques. Let
$(A,B,S^*)$ be a proper decomposition of $\graph$, where $S^*\in \sep$ is a minimal
separator of $\graph$. Let $\cliques_1,\sep_1$ denote the cliques and
separators of $\graph_{A\cup S^*}$, with separator multiplicities $\nu_1$, and
similarly let $\cliques_2,\sep_2,\nu_2$ correspond to $\graph_{B\cup S^*}$. Strong subadditivity gives, for every $\rho\in\mathbb S_1^+(\hilb_V)$,
\begin{equation}\label{eq:subadd-dec}
{\rm S}(\rho)\le {\rm S}(\rho_{A\cup S^*})+{\rm S}(\rho_{B\cup S^*})-{\rm S}(\rho_{S^*}).
\end{equation}
Applying the induction hypothesis to $\rho_{A\cup S^*}$ and $\rho_{B\cup S^*}$
yields
\begin{align*}
{\rm S}(\rho)
&\le {\rm S}(\rho_{A\cup S^*})+{\rm S}(\rho_{B\cup S^*})-{\rm S}(\rho_{S^*}) \\
&\le \sum_{C\in\cliques_1}{\rm S}(\rho_C)-\sum_{S\in\sep_1}\nu_1(S){\rm S}(\rho_S)
 + \sum_{C\in\cliques_2}{\rm S}(\rho_C)-\sum_{S\in\sep_2}\nu_2(S){\rm S}(\rho_S)-{\rm S}(\rho_{S^*}) \\
&= \sum_{C\in\cliques}{\rm S}(\rho_C)-\sum_{S\in\sep}\nu(S){\rm S}(\rho_S).
\end{align*}
This proves the inequality.

Suppose now that equality holds. Then equality must hold in
\eqref{eq:subadd-dec} and in the two applications of the induction hypothesis.
Hence $\rho_{A\cup S^*}$ and $\rho_{B\cup S^*}$ are quantum Markov on the
induced subgraphs, and $A\ciph B\cd S^*\,[\rho]$. By
Lemma~\ref{lem:decomp}, this implies that $\rho$ is quantum Markov on
$\graph$. Conversely, if $\rho$ is quantum Markov on $\graph$, then
Lemma~\ref{lem:decomp} implies that $\rho_{A\cup S^*}$ and
$\rho_{B\cup S^*}$ are quantum Markov on the induced subgraphs and that
$A\ciph B\cd S^*\,[\rho]$. By the induction hypothesis, equality holds in the
two induced-graph entropy bounds, and by strong subadditivity equality also
holds in \eqref{eq:subadd-dec}. Therefore equality holds in the displayed
entropy inequality for $\graph$.
\end{proof}

For later reference we introduce the entropy defect associated with the chordal
Markov entropy formula which we shall term the \emph{global information} in $\rho$ relative to the graph $\graph$.
\begin{equation}\label{eq:gI}
g{\rm I}(\graph)_\rho\;:=\;
\sum_{C\in\cliques}{\rm S}(\rho_C)
-
\sum_{S\in\sep}\nu(S){\rm S}(\rho_S)
-
{\rm S}(\rho).
\end{equation}
By Proposition~\ref{prop:ssa-chordal}, $g{\rm I}(\graph)_\rho\geq 0$, with equality if and
only if $\rho$ is quantum Markov with respect to $\graph$. Thus
$g{\rm I}(\graph)_\rho$ measures the entropy defect, or global information, not
captured by the clique and separator reductions specified by $\graph$.

This quantity interpolates between familiar information measures. If $\graph$
is complete, then $g{\rm I}(\graph)_\rho=0$ for every state $\rho$. If $\graph$ is
empty, then
$$
g{\rm I}(\graph)_\rho=\sum_{v\in V}{\rm S}(\rho_v)-{\rm S}(\rho),
$$
the quantum analogue of Watanabe's total correlation, or multiinformation
\citep{watanabe1960information,studeny:05}. In the two-clique case, with
cliques $A\cup S$ and $B\cup S$, \eqref{eq:gI} reduces to the conditional
mutual information ${\rm I}(A:B\cd S)_\rho$.

This maximum entropy viewpoint also connects $g{\rm I}(\graph)_\rho$ to notions of
connected information. In the classical theory, irreducible higher-order
correlations are defined by comparing a distribution with the maximum entropy
distribution compatible with prescribed lower-order marginals
\citep{schneidman2003network}. Quantum analogues of irreducible multiparty
correlations are defined similarly, using maximum entropy states compatible
with prescribed reduced density operators \citep{zhou2008irreducible}. The
quantity $g{\rm I}(\graph)_\rho$ follows the same philosophy, but organizes the
prescribed local data through the cliques of a chordal graph. The divergence
identity in Proposition~\ref{prop:gI-divergence-identity} below makes this
interpretation precise.

\color{black}

\subsection{The chordal logarithmic construction}

We
now turn to clique marginals on a chordal graph. The point is not merely that
the same logarithmic construction extends formally, but that the two main
features of the two-clique case survive: the trace of $T(\mathcal R)$ still
detects the existence of a Markov completion, and whenever such a completion
exists it is unique and equal to $T(\mathcal R)$.

\begin{thm}[Chordal trace criterion]\label{thm:chordal-trace-criterion}
Let $\graph$ be chordal with clique set $\cliques$, and let
$\mathcal R=\{\rho_C:C\in\cliques\}$ be a pairwise consistent family of
strictly positive clique marginals. Define
\begin{equation}\label{eq:logT-decomposable}
T(\mathcal R)
=
\exp\left\{
\sum_{C\in\cliques}\log\rho_C
-
\sum_{S\in\sep}\nu(S)\log\rho_S
\right\}.
\end{equation}
Then $\trace (T(\mathcal R))\le 1$. Moreover, the following are equivalent:
\begin{enumerate}[(i)]
\item $\trace (T(\mathcal R))=1$;
\item $T(\mathcal R)\in M(\mathcal R)$;
\item there exists $\omega\in M(\mathcal R)$ that is quantum Markov on
      $\graph$.
\end{enumerate}
When these conditions hold, $T(\mathcal R)$ is the unique quantum Markov
completion and the unique maximum entropy element of $M(\mathcal R)$. Moreover,
$$
{\rm S}(T(\mathcal R))
=
\sum_{C\in\cliques}{\rm S}(\rho_C)
-
\sum_{S\in\sep}\nu(S){\rm S}(\rho_S).
$$
\end{thm}

The proof uses the following variational identity, the chordal analogue of
Lemma~\ref{lem:two-clique-divergence-identity-general}.

\begin{lem}\label{lem:chordal_variational}
If $M(\mathcal R)\neq\emptyset$, then for every $\omega\in M(\mathcal R)$,
\begin{equation}\label{eq:rel_entropy_dec}
{\rm D}(\omega\|T(\mathcal R))+1-\trace(T(\mathcal R))
=
\sum_{C\in\cliques}{\rm S}(\rho_C)
-
\sum_{S\in\sep}\nu(S){\rm S}(\rho_S)
-
{\rm S}(\omega).
\end{equation}
\end{lem}

\begin{proof}
Let $\omega\in M(\mathcal R)$. Using \eqref{eq:logT-decomposable}, we get
\begin{align*}
{\rm D}(\omega\|T(\mathcal R))
&=
\trace(\omega\log\omega)-\trace(\omega\log T(\mathcal R))
+\trace(T(\mathcal R))-1 \\
&=
\trace(\omega\log\omega)
-\sum_{C\in\cliques}\trace(\omega_C\log\rho_C)
+\sum_{S\in\sep}\nu(S)\trace(\omega_S\log\rho_S)
+\trace(T(\mathcal R))-1.
\end{align*}
Since $\omega\in M(\mathcal R)$, we have $\omega_C=\rho_C$ for every
$C\in\cliques$, and by consistency also $\omega_S=\rho_S$ for every
$S\in\sep$. Therefore
$$
{\rm D}(\omega\|T(\mathcal R))+1-\trace(T(\mathcal R))
=
-{\rm S}(\omega)
+
\sum_{C\in\cliques}{\rm S}(\rho_C)
-
\sum_{S\in\sep}\nu(S){\rm S}(\rho_S),
$$
which is \eqref{eq:rel_entropy_dec}.
\end{proof}

\begin{proof}[Proof of Theorem~\ref{thm:chordal-trace-criterion}]
Set $c=\trace(T(\mathcal R))$ and $\sigma=T(\mathcal R)/c$. Then $\sigma$ is
a density operator and $\log T(\mathcal R)=(\log c)I+\log\sigma$. Multiplying
by $\sigma$ and taking the trace gives
$$
\log c
=
\trace(\sigma\log T(\mathcal R))+{\rm S}(\sigma).
$$
Using \eqref{eq:logT-decomposable}, we obtain
$$
\log c
=
\sum_{C\in\cliques}\trace(\sigma_C\log\rho_C)
-
\sum_{S\in\sep}\nu(S)\trace(\sigma_S\log\rho_S)
+
{\rm S}(\sigma).
$$
Since $\trace(\tau\log\eta)=-{\rm D}(\tau\|\eta)-{\rm S}(\tau)$ for density operators, this
becomes
\begin{align}
\log c
&=
-\Bigl[
\sum_{C\in\cliques}{\rm D}(\sigma_C\|\rho_C)
-
\sum_{S\in\sep}\nu(S){\rm D}(\sigma_S\|\rho_S)
\Bigr] \nonumber\\
&\qquad
+
\Bigl[
{\rm S}(\sigma)-\sum_{C\in\cliques}{\rm S}(\sigma_C)
+\sum_{S\in\sep}\nu(S){\rm S}(\sigma_S)
\Bigr].
\label{eq:trace-chordal-proof}
\end{align}
By Proposition~\ref{prop:ssa-chordal}, the second bracket is nonpositive. We
claim that the first bracket is nonnegative.

The claim is proved by induction on the number of cliques. If there is only one
clique, it is immediate. For the induction step, let $C_0$ be a leaf clique in
a junction tree and let $S_0$ be the unique separator connecting $C_0$ to the
rest of the tree. Removing $C_0$ gives a smaller chordal graph with clique set
$\cliques'$ and separator multiplicities $\nu'$. Then
\begin{align*}
&\sum_{C\in\cliques}{\rm D}(\sigma_C\|\rho_C)
-
\sum_{S\in\sep}\nu(S){\rm D}(\sigma_S\|\rho_S) \\
&\qquad =
\bigl({\rm D}(\sigma_{C_0}\|\rho_{C_0})-{\rm D}(\sigma_{S_0}\|\rho_{S_0})\bigr)
+
\sum_{C\in\cliques'}{\rm D}(\sigma_C\|\rho_C)
-
\sum_{S\in\sep'}\nu'(S){\rm D}(\sigma_S\|\rho_S).
\end{align*}
The first term is nonnegative by monotonicity of relative entropy under partial
trace, and the second is nonnegative by the induction hypothesis. This proves
the claim. Hence $\log c\le 0$, and therefore $\trace(T(\mathcal R))\le 1$.

We now prove the equivalences. Suppose first that $\trace(T(\mathcal R))=1$.
Then $c=1$ and $\sigma=T(\mathcal R)$. In \eqref{eq:trace-chordal-proof}, the
left-hand side is zero, while both terms on the right are nonpositive. Hence
both brackets vanish. In particular,
$$
\sum_{C\in\cliques}{\rm D}(\sigma_C\|\rho_C)
-
\sum_{S\in\sep}\nu(S){\rm D}(\sigma_S\|\rho_S)
=
0.
$$
We now argue that this forces $\sigma_C=\rho_C$ for every $C\in\cliques$. The proof
is again by induction on the number of cliques. If there is only one clique,
the displayed equality gives ${\rm D}(\sigma_C\|\rho_C)=0$, hence
$\sigma_C=\rho_C$. For the induction step, use the same leaf clique $C_0$ and
separator $S_0$ as above. The same decomposition writes the vanishing expression
as a sum of two nonnegative terms. Therefore both terms vanish. By the induction
hypothesis applied to the smaller chordal graph, $\sigma_C=\rho_C$ for all
$C\in\cliques'$. In particular, $\sigma_{S_0}=\rho_{S_0}$. Since
$$
{\rm D}(\sigma_{C_0}\|\rho_{C_0})-{\rm D}(\sigma_{S_0}\|\rho_{S_0})=0,
$$
we get ${\rm D}(\sigma_{C_0}\|\rho_{C_0})=0$, and hence
$\sigma_{C_0}=\rho_{C_0}$. Thus $\sigma_C=\rho_C$ for all cliques $C$. Since
$c=1$, this means $T(\mathcal R)\in M(\mathcal R)$. Hence
$(\mathrm{i})\Rightarrow(\mathrm{ii})$. The implication
$(\mathrm{ii})\Rightarrow(\mathrm{i})$ is immediate, since every element of
$M(\mathcal R)$ has trace one.

Next assume $T(\mathcal R)\in M(\mathcal R)$. Lemma~\ref{lem:chordal_variational}
gives
$$
{\rm D}(T(\mathcal R)\|T(\mathcal R))+1-\trace(T(\mathcal R))
=
\sum_{C\in\cliques}{\rm S}(\rho_C)
-
\sum_{S\in\sep}\nu(S){\rm S}(\rho_S)
-
{\rm S}(T(\mathcal R)).
$$
The left-hand side is zero. Hence
$$
{\rm S}(T(\mathcal R))
=
\sum_{C\in\cliques}{\rm S}(\rho_C)
-
\sum_{S\in\sep}\nu(S){\rm S}(\rho_S).
$$
By the equality characterization in Proposition~\ref{prop:ssa-chordal},
$T(\mathcal R)$ is quantum Markov on $\graph$. Thus
$(\mathrm{ii})\Rightarrow(\mathrm{iii})$.

Finally, suppose that $\omega\in M(\mathcal R)$ is quantum Markov on $\graph$.
Then $\omega_C=\rho_C$ for every $C\in\cliques$, and by consistency also
$\omega_S=\rho_S$ for every $S\in\sep$. Since $\omega$ is quantum Markov,
Proposition~\ref{prop:ssa-chordal} gives
$$
{\rm S}(\omega)
=
\sum_{C\in\cliques}{\rm S}(\rho_C)
-
\sum_{S\in\sep}\nu(S){\rm S}(\rho_S).
$$
Applying Lemma~\ref{lem:chordal_variational} gives
$$
{\rm D}(\omega\|T(\mathcal R))+1-\trace(T(\mathcal R))=0.
$$
Both terms are nonnegative: the first by Klein's inequality  and the second by the trace bound proved above.
Hence both terms vanish. In particular, ${\rm D}(\omega\|T(\mathcal R))=0$ and
$\trace(T(\mathcal R))=1$. Thus $\omega=T(\mathcal R)$. This proves
$(\mathrm{iii})\Rightarrow(\mathrm{i})$ and also uniqueness of the Markov
completion.

It remains to show the maximum entropy statement. When the equivalent
conditions hold, $T(\mathcal R)\in M(\mathcal R)$ and
$$
\log T(\mathcal R)
=
\sum_{C\in\cliques}\log\rho_C
-
\sum_{S\in\sep}\nu(S)\log\rho_S.
$$
Thus $\log T(\mathcal R)$ lies in the linear span of the local clique and
separator operator spaces. By Theorem~\ref{thm:maxent_characterization},
$T(\mathcal R)$ is the unique maximum entropy element of $M(\mathcal R)$.

Finally, the entropy formula follows from the identity already proved in the
implication $(\mathrm{ii})\Rightarrow(\mathrm{iii})$.
\end{proof}

The variational identity also identifies the global information
$g{\rm I}(\graph)_\rho$. Recall that $g{\rm I}(\graph)_\rho$ is the difference between the
chordal entropy expression determined by the clique marginals of $\rho$ and the
entropy of $\rho$ itself. The next proposition shows that this gap is the
relative entropy from $\rho$ to the logarithmic candidate, with the same trace
correction as above.

\begin{prop}[Global information as divergence]\label{prop:gI-divergence-identity}
Let $\rho\in\mathbb S_1^+(\hilb_V)$ and
$\mathcal R_\rho=\{\rho_C:C\in\cliques\}$. Then
$$
g{\rm I}(\graph)_\rho
=
{\rm D}(\rho\|T(\mathcal R_\rho))+1-\trace(T(\mathcal R_\rho)).
$$
In particular, if $\trace(T(\mathcal R_\rho))=1$ and
$\rho^\star_\graph=T(\mathcal R_\rho)$, then
$$
g{\rm I}(\graph)_\rho
=
{\rm D}(\rho\|\rho^\star_\graph)
=
{\rm S}(\rho^\star_\graph)-{\rm S}(\rho).
$$
\end{prop}

\begin{proof}
Apply Lemma~\ref{lem:chordal_variational} with
$\mathcal R=\mathcal R_\rho$ and $\omega=\rho$. Since $\rho$ has clique
marginals $\mathcal R_\rho$, the right-hand side of
\eqref{eq:rel_entropy_dec} is precisely $g{\rm I}(\graph)_\rho$. This gives the
first identity. If $\trace(T(\mathcal R_\rho))=1$, then
$\rho^\star_\graph=T(\mathcal R_\rho)$ is a density operator and the trace
correction vanishes. The identity
$g{\rm I}(\graph)_\rho={\rm S}(\rho^\star_\graph)-{\rm S}(\rho)$ follows from the entropy formula
in Theorem~\ref{thm:chordal-trace-criterion}.
\end{proof}

When the trace-one condition holds, the global information has the same form as
in the classical chordal case: it is the relative entropy from the state to the
canonical completion determined by its clique marginals. When the trace is
strictly smaller than one, the same formula retains a correction term recording
the failure of the logarithmic candidate to be normalized.

The examples in Section~\ref{sec:examples} show that this correction term is not
only formal. In Example~\ref{ex:basic-qubit}, the prescribed two-body marginals
are strictly feasible when $\varepsilon^2+\delta^2<1$, but
Lemma~\ref{lem:basic-qubit-trace-defect} shows that
$\trace(T(\mathcal R))<1$ whenever $\varepsilon\delta\ne0$. Thus the logarithmic
candidate fails to be normalized exactly in the genuinely noncommuting case.

\subsection{Maximum entropy completions}\label{sec:maxent}

The preceding results identify when the logarithmic candidate
$T(\mathcal R)$ is feasible and Markov in the case when $\mathcal R$ is the set  of cliques of a chordal graph. We now investigate the general
maximum entropy principle behind such logarithmic formulae. This result is not
specific to chordal graphs. For any feasible family of prescribed marginals,
the entropy maximizer is characterized by a dual condition: its logarithm lies
in the linear span of the local operator spaces that define the constraints. We emphasize that  even in the classical situation there is no simple condition for $M(\cR)$ to be non-empty unless $\gen$ is the set of cliques of a chordal graph. So here we describe properties of the entropy maximizer under the condition that a completion exists.

The geometry is the usual primal--dual geometry of relative entropy
\citep{chentsov:68,chentsov:72,csiszar:75}. Density operators form the primal affine space, while logarithms of density
operators give the corresponding dual coordinates. Thus
$$
\langle \rho-\sigma,\log\sigma-\log\tau\rangle=\trace((\rho-\sigma)(\log\sigma-\log\tau))
$$
is the natural pairing between a primal displacement $\rho-\sigma$ and a dual
displacement $\log\sigma-\log\tau$. The vanishing of this pairing is the
orthogonality condition behind the Pythagorean identity for relative entropy.

\begin{lem}[Three-point identity]\label{lem:three-point-identity}
Let $\rho,\sigma,\tau\in\mathbb S_1^+(\hilb_V)$. Then
$$
{\rm D}(\rho\|\tau)
=
{\rm D}(\rho\|\sigma)+{\rm D}(\sigma\|\tau)
+
\langle \rho-\sigma,\log\sigma-\log\tau\rangle.
$$
Equivalently, if $\rho_\alpha=\alpha\rho+(1-\alpha)\sigma$, then
$$
\left.\frac{d}{d\alpha}{\rm D}(\rho_\alpha\|\tau)\right|_{\alpha=0}
=
{\rm D}(\rho\|\tau)-{\rm D}(\rho\|\sigma)-{\rm D}(\sigma\|\tau).
$$
In particular, if
$\langle \rho-\hat\sigma,\log\hat\sigma-\log\tau\rangle=0$, then we have a Pythagorean identity
$$
{\rm D}(\rho\|\tau)={\rm D}(\rho\|\hat\sigma)+{\rm D}(\hat\sigma\|\tau).
$$
\end{lem}

\begin{proof}
Expanding the three relative entropies gives
\begin{align*}
&{\rm D}(\rho\|\tau)-{\rm D}(\rho\|\sigma)-{\rm D}(\sigma\|\tau) \\
&\qquad =
\trace(\rho\log\sigma)-\trace(\rho\log\tau)
-\trace(\sigma\log\sigma)+\trace(\sigma\log\tau)  \\
&\qquad =
\langle \rho-\sigma,\log\sigma-\log\tau\rangle.
\end{align*}
This proves the first identity. For the derivative formula, write
$\rho_\alpha=\sigma+\alpha(\rho-\sigma)$. We use the standard directional
derivative
$$
\left.\frac{d}{dt}\trace\{(X+tH)\log(X+tH)\}\right|_{t=0}
=
\trace\{H(\log X+I)\},
$$
valid for strictly positive $X$ and self-adjoint $H$. Applying this with
$X=\sigma$ and $H=\rho-\sigma$, we get
\begin{align*}
\left.\frac{d}{d\alpha}{\rm D}(\rho_\alpha\|\tau)\right|_{\alpha=0}
&=
\trace\{(\rho-\sigma)(\log\sigma+I-\log\tau)\} =
\langle \rho-\sigma,\log\sigma-\log\tau\rangle,
\end{align*}
because $\trace(\rho-\sigma)=0$. The derivative formula now follows from the
first identity. The Pythagorean identity is the special case where the final
pairing vanishes.
\end{proof}

The next theorem applies this geometry to an arbitrary family of prescribed
marginals. The feasible set $M(\mathcal R)$ is an affine subset of the primal
space of density operators. Its maximum entropy element is characterized by
the dual condition that its logarithm is orthogonal to all feasible directions,
or equivalently lies in the span of the local constraint operators.

\begin{thm}\label{thm:maxent_characterization}
Let $\gen$ be a family of subsets of $V$ and 
$\mathcal R=\{\rho_A: A\in\gen\}$ a feasible family of prescribed
marginals. For $\hat\rho\in M(\mathcal R)$, the following are equivalent:
\begin{enumerate}[(i)]
\item $\hat\rho$ is the unique maximum entropy element of $M(\mathcal R)$;
\item $\hat\rho$ has a log-linear expansion of the form
$$
\log\hat\rho
=
\lambda I+\sum_{A\in\gen}(M_A\otimes I_{V\setminus A})
$$
for some $\lambda\in\R$ and self-adjoint operators
$M_A\in\mathbb S(\hilb_A)$.
\end{enumerate}
\end{thm}

\begin{proof}
Consider the space of feasible directions that keep the prescribed marginals
fixed,
$$
\mathcal V
=
\{H\in\mathbb S(\hilb_V):\trace(H)=0,\ H_A=0\text{ for all }A\in\gen\}.
$$
We first identify its orthogonal complement under the trace pairing. We claim
that
$$
\mathcal V^\perp
=
\left\{
\lambda I+\sum_{A\in\gen}(M_A\otimes I_{V\setminus A}):
\lambda\in\R,\ M_A\in\mathbb S(\hilb_A)
\right\}.
$$
Every operator on the right is orthogonal to every $H\in\mathcal V$, because
$$
\langle
H,\lambda I+\sum_{A\in\gen}(M_A\otimes I_{V\setminus A})
\rangle
=
\lambda\trace(H)+\sum_{A\in\gen}\trace(H_A M_A)
=
0.
$$
Conversely, define the linear map
$$
\Phi:\mathbb S(\hilb_V)\longrightarrow
\R\oplus\bigoplus_{A\in\gen}\mathbb S(\hilb_A),
\qquad
\Phi(H)=(\trace(H),(H_A)_{A\in\gen}).
$$
Then $\mathcal V=\ker(\Phi)$. The elementary finite-dimensional identity
$(\ker(\Phi))^\perp=\operatorname{im}(\Phi^*)$ gives the reverse inclusion. Indeed,
the adjoint map satisfies
$$
\Phi^*(\lambda,(M_A)_{A\in\gen})
=
\lambda I+\sum_{A\in\gen}(M_A\otimes I_{V\setminus A}),
$$
since
$$
\left\langle H,\Phi^*(\lambda,(M_A)_{A\in\gen})\right\rangle
=
\lambda\trace(H)+\sum_{A\in\gen}\trace(H_A M_A).
$$
This proves the displayed formula for $\mathcal V^\perp$.

Suppose now that $\hat\rho\in M(\mathcal R)$ has the log-linear form in (ii). By the formula
for $\mathcal V^\perp$, this means $\log\hat\rho\in\mathcal V^\perp$. Let
$\omega\in M(\mathcal R)$. Then $\omega-\hat\rho\in\mathcal V$, and hence
$\langle \omega-\hat\rho,\log\hat\rho\rangle=0$.
Since also $\trace(\omega-\hat\rho)=0$, this is equivalent to
$$
\langle \omega-\hat\rho,\log\hat\rho-\log(I/d)\rangle=0,
\qquad d=\dim\hilb_V.
$$
The Pythagorean identity in Lemma~\ref{lem:three-point-identity}, with
$\rho=\omega$, $\sigma=\hat\rho$, and $\tau=I/d$, gives
$$
{\rm D}(\omega\|I/d)
=
{\rm D}(\omega\|\hat\rho)+{\rm D}(\hat\rho\|I/d).
$$
Since ${\rm D}(\eta\|I/d)=\log d-{\rm S}(\eta)$, this is equivalent to
${\rm S}(\hat\rho)-{\rm S}(\omega)={\rm D}(\omega\|\hat\rho)\ge 0$.
Equality holds if and only if $\omega=\hat\rho$. Thus $\hat\rho$ is the unique
maximum entropy element of $M(\mathcal R)$; proving (ii)$\Rightarrow$(i).

Conversely, suppose that $\hat\rho$ is the maximum entropy element of
$M(\mathcal R)$. For every $H\in\mathcal V$ and all sufficiently small real
$\epsilon$, the operator $\hat\rho+\epsilon H$ remains in the same affine
marginal space. The first-order optimality condition gives
$$
0
=
\left.\frac{d}{d\epsilon}{\rm S}(\hat\rho+\epsilon H)\right|_{\epsilon=0}
=
-\trace\{H(\log\hat\rho+I)\}.
$$
Since $\trace(H)=0$, this gives
$
\langle H,\log\hat\rho\rangle=0$ for all $H\in\mathcal V$.
Hence $\log\hat\rho\in\mathcal V^\perp$, which is exactly the log-linear form
in (ii).
\end{proof}

The proof gives the useful identity
$$
{\rm S}(\hat\rho)-{\rm S}(\omega)={\rm D}(\omega\|\hat\rho),
\qquad \omega\in M(\mathcal R).
$$
Thus every other element in $M(\mathcal R)$ has smaller entropy than $\hat \rho$, and the entropy gap is exactly
the relative entropy from $\omega$ to the maximum entropy completion
$\hat\rho$. This also gives uniqueness, since the gap is zero only when
$\omega=\hat\rho$.

In the chordal setting, the log-linear form in
Theorem~\ref{thm:maxent_characterization} specializes to the operator
$T(\mathcal R)$ studied above. Hence the trace-one condition in
Theorem~\ref{thm:chordal-trace-criterion} identifies exactly when this formal
logarithmic candidate is normalized, feasible, Markov, and the
maximum entropy completion.

\begin{rem}
For classical positive distributions on a chordal graph, every pairwise
consistent family of clique marginals has a Markov extension, given by the
junction-tree formula, and this extension is also the maximum entropy
completion. The quantum situation is different: the maximum entropy completion
may exist and have log-linear form without being quantum Markov. This failure is illustrated in Example~\ref{ex:basic-qubit} and
Lemma~\ref{lem:basic-qubit-no-markov}; see also the
remarks after Theorem~4.7 in \cite{leifer_poulin_2008}.
\end{rem}

\section{Examples and counterexamples}\label{sec:examples}

The preceding results show that the quantum marginal problem is controlled by
phenomena with no classical analogue. We now illustrate these obstructions
using elementary Pauli expansions. The first examples are on three qubits and
separate local consistency, feasibility, Markov feasibility, and maximum
entropy. The final example uses the chordal graph in Figure~\ref{fig:chordal}
and illustrates the chordal trace criterion beyond the three-chain case.
\color{black}

The examples separate several properties that coincide, or nearly coincide, in
the classical decomposable case:
\begin{enumerate}[(i)]
\item local consistency of overlapping marginals;
\item feasibility, that is, existence of some global completion;
\item Markov feasibility, that is, existence of a quantum conditionally
independent completion;
\item maximum entropy completion.
\end{enumerate}
They illustrate that local consistency need not imply feasibility; feasibility need
not imply Markov feasibility; and the maximum entropy completion need not be
Markov.

\subsection{Pauli preliminaries}

We use the standard Pauli matrices
$$
I=\begin{pmatrix}1&0\\0&1\end{pmatrix},\quad
X=\begin{pmatrix}0&1\\1&0\end{pmatrix},\quad
Y=\begin{pmatrix}0&-i\\ i&0\end{pmatrix},\quad
Z=\begin{pmatrix}1&0\\0&-1\end{pmatrix}.
$$
They are Hermitian, satisfy $X^2=Y^2=Z^2=I$, and obey
$$
XY=-YX=iZ,\qquad
YZ=-ZY=iX,\qquad
ZX=-XZ=iY.
$$
With respect to the trace pairing $\langle A,B\rangle=\trace(AB)$ on
self-adjoint matrices, the four matrices $I,X,Y,Z$ form an orthogonal basis (over $\R$) of
$\mathbb S(\C^2)$.

On $(\C^2)^{\otimes n}$, a \emph{Pauli word} is a tensor product
$$
W=P_1\otimes\cdots\otimes P_n,
\qquad
P_j\in\{I,X,Y,Z\}.
$$
Each Pauli word is Hermitian and satisfies $W^2=I$. The $4^n$ Pauli words form
an orthogonal basis of $\mathbb S((\C^2)^{\otimes n})$.
We shall repeatedly use two elementary facts. Firstly, if
$$
\rho=\frac1{2^n}\left(I+\sum_W a_WW\right),
$$
then tracing out a tensor factor annihilates every Pauli word acting nontrivially on
that factor. Secondly, pairwise anticommuting Pauli words behave like orthogonal
Euclidean directions.

\begin{lem}\label{lem:Pauli-ball}
Let $W_1,\dots,W_m$ be pairwise anticommuting Pauli words. Set
$$
T=\sum_{i=1}^m a_iW_i,
\qquad
r=\left(\sum_i a_i^2\right)^{1/2},
\qquad
\rho=\frac1{2^n}(I+T).
$$
Then
$$
T^2=r^2I.
$$
Consequently, $\rho\succ0$ if and only if $r<1$.
\end{lem}

\begin{proof}
Anticommutation cancels all cross-terms in $T^2$. Thus $T^2=r^2I$, so the
eigenvalues of $T$ are $\pm r$. Hence the eigenvalues of $\rho$ are
$2^{-n}(1\pm r)$.
\end{proof}

\begin{lem}\label{lem:Pauli-spectral}
Under the assumptions of Lemma~\ref{lem:Pauli-ball}, suppose $r>0$. Then
$$
\Pi_\pm=\frac12\left(I\pm\frac1rT\right)
$$
are the spectral projections of $T$, and
$$
\rho
=
\frac{1+r}{2^n}\Pi_+
+
\frac{1-r}{2^n}\Pi_-.
$$
If moreover $r<1$, then
\begin{equation}\label{eq:Pauli-log-formula}
\log\rho
=
\log\left(\frac{\sqrt{1-r^2}}{2^n}\right)I
+
\frac{\operatorname{arctanh}(r)}{r}\,T.
\end{equation}
\end{lem}

\begin{proof}
Everything follows immediately with some algebra from $T^2=r^2I$ and the fact that $\operatorname{arctanh}(r)=\tfrac12\log(\tfrac{1+r}{1-r})$.
\end{proof}

\subsection{Local consistency need not imply feasibility}
This fact is well known, but we illustrate the fact with a simple example.
\begin{ex}\label{ex:basic-qubit}
Let $\hilb=(\C^2)^{\otimes3}$ and define
$$
W_1=X\otimes X\otimes I,
\qquad
W_2=I\otimes Z\otimes Z.
$$
The words $W_1$ and $W_2$ anticommute, because they anticommute on the middle
tensor factor. For $|\varepsilon|,|\delta|<1$, prescribe
$$
\rho_{12}=\frac14(I+\varepsilon X\otimes X),
\qquad
\rho_{23}=\frac14(I+\delta Z\otimes Z).
$$
Both marginals are strictly positive. They are also consistent on the overlap,
since $(\rho_{12})_2=(\rho_{23})_2=I/2$.
\end{ex}

\begin{lem}\label{lem:basic-qubit-feasibility}
In Example~\ref{ex:basic-qubit}, the prescribed marginals are feasible if
$\varepsilon^2+\delta^2\le1$, and are not feasible if
$\varepsilon^2+\delta^2>1$.
\end{lem}

\begin{proof}
If $\varepsilon^2+\delta^2\le1$, then
$$
\sigma=\frac18(I+\varepsilon W_1+\delta W_2)
$$
is positive semidefinite by Lemma~\ref{lem:Pauli-ball}. Its marginals are
$\rho_{12}$ and $\rho_{23}$. Conversely, let $\omega$ be any completion of $\rho_{12},\rho_{23}$. Then
$$
\trace(\omega W_1)=\trace(\trace_3(\omega) X\otimes X)=\tfrac14\trace((I+\varepsilon X\otimes X) X\otimes X)=\varepsilon.$$ Similarly,
$\trace(\omega W_2)=\delta$.
Set
$
T=\varepsilon W_1+\delta W_2.
$
Since $T^2=(\varepsilon^2+\delta^2)I$,
$
\|T\|=(\varepsilon^2+\delta^2)^{1/2}.
$
Hence, using the fact that $T\preceq \|T\| I$ and $\trace(\omega)=1$,
$$
\varepsilon^2+\delta^2
=
\trace(\omega T)
\le
\|T\|
=
(\varepsilon^2+\delta^2)^{1/2},
$$
which is impossible when $\varepsilon^2+\delta^2>1$.
\end{proof}

Thus local consistency does not imply feasibility. Moreover, on the boundary
$\varepsilon^2+\delta^2=1$, every completion is singular.

\subsection{Feasibility need not imply Markov feasibility}
Again, this fact is also well known, but for illustration we provide a simple example that feasibility does
not imply the existence of a quantum Markov completion.
We now stay in the strictly feasible regime
$\varepsilon^2+\delta^2<1$.  Set
$$
K
=
\rho_{12}^{1/2}\rho_2^{-1/2}\rho_{23}^{1/2}.
$$
Since $\rho_2=I/2$,
$$
K=\sqrt2\,\rho_{12}^{1/2}\rho_{23}^{1/2}.
$$
By the characterization of quantum Markov states, a Markov completion would
have to equal both $KK^\ast$ and $K^\ast K$. Thus $K$ must be normal.

\begin{lem}\label{lem:normal-product}
Let $A,B\succ0$. Then $A^{1/2}B^{1/2}$ is normal if and only if $AB=BA$.
\end{lem}

\begin{proof}
The operator $A^{1/2}B^{1/2}$ is similar to
$B^{1/4}A^{1/2}B^{1/4}$, which is positive definite. Hence, if
$A^{1/2}B^{1/2}$ is normal, it is unitarily diagonalizable with positive
spectrum (in particular, real), and so it is self-adjoint. Thus
$A^{1/2}B^{1/2}=B^{1/2}A^{1/2}$, equivalently $AB=BA$. The converse is
immediate.
\end{proof}

Applying Lemma~\ref{lem:normal-product} with
$A=\rho_{12}$ and $B=\rho_{23}$, viewed as operators on
$\mathcal H_1\otimes\mathcal H_2\otimes\mathcal H_3$, we conclude that
$K$ is normal if and only if $\rho_{12}$ and $\rho_{23}$ commute. In the present example,
$
\rho_{12}=\frac14(I+\varepsilon W_1)$ and $
\rho_{23}=\frac14(I+\delta W_2)$,
and $W_1W_2=-W_2W_1$. The commutator is
$$
[\rho_{12},\rho_{23}]
=
\frac{\varepsilon\delta}{16}[W_1,W_2].
$$
Therefore $\rho_{12}$ and $\rho_{23}$ commute if and only if
$\varepsilon\delta=0$.

The following result follows immediately from Proposition~\ref{prop:petz2} and the lemma above.

\begin{lem}\label{lem:basic-qubit-no-markov}
Assume $\varepsilon^2+\delta^2<1$ in
Example~\ref{ex:basic-qubit}. Then the prescribed marginals admit a quantum
Markov completion if and only if $\varepsilon\delta=0$.
\end{lem}

The same obstruction appears in the logarithmic candidate
$T(\mathcal R)$.

\begin{lem}\label{lem:basic-qubit-trace-defect}
Let $\mathcal R=\{\rho_{12},\rho_{23}\}$ be the marginal family from
Example~\ref{ex:basic-qubit}, and assume $\varepsilon^2+\delta^2<1$. Set
$$
b=\operatorname{arctanh}(\varepsilon),
\qquad
d=\operatorname{arctanh}(\delta),
\qquad
r=(b^2+d^2)^{1/2}.
$$
Then
$$
\trace(T(\mathcal R))
=
\sqrt{(1-\varepsilon^2)(1-\delta^2)}\,\cosh(r).
$$
Thus,
$
\trace(T(\mathcal R))=1$ if and only if $\varepsilon\delta=0.
$
\end{lem}

\begin{proof}
By \eqref{eq:Pauli-log-formula},
$$
\log\rho_{12}
=
\log(\tfrac{\sqrt{1-\varepsilon^2}}{4})I+bW_1,\qquad
\log\rho_{23}
=
\log(\tfrac{\sqrt{1-\delta^2}}{4})I+dW_2.
$$
Since $\rho_2=I/2$, the exponent defining $T(\mathcal R)$ equals
$$\log(\tfrac{\sqrt{(1-\varepsilon^2)(1-\delta^2)}}{8})
I
+
bW_1+dW_2.
$$
Because $W_1$ and $W_2$ anticommute,
$
(bW_1+dW_2)^2=r^2I$.
Hence
$$
e^{bW_1+dW_2}
=
\cosh(r)I+\frac{\sinh(r)}{r}(bW_1+dW_2).
$$
Taking traces gives the formula. Using
$\varepsilon=\tanh b$ and $\delta=\tanh d$,
$$
\trace(T(\mathcal R))
=
\frac{\cosh r}{\cosh b\,\cosh d}.
$$
If $bd=0$ (equiv. $\varepsilon\delta=0$), this equals one. If $bd\ne0$, concavity of
$f(t)=\log \cosh\sqrt t$ and the fact that $f(0)=1>0$ imply that $f$ is strictly subadditive so
$
\cosh r
<
\cosh b\,\cosh d
$.
Hence $\trace(T(\mathcal R))<1$.
\end{proof}

The feasible completion
$$
\sigma=\frac18(I+\varepsilon W_1+\delta W_2)
$$
is also the unique maximum entropy completion. Indeed,
$\log\sigma$ lies in the span of
$I,W_1,W_2$, where $W_1$ is supported on subsystem $12$ and $W_2$ on subsystem
$23$. Thus $\sigma$ has the required log-linear form. By
Lemma~\ref{lem:basic-qubit-no-markov}, this maximum entropy completion
fails to be Markov whenever $\varepsilon\delta\ne0$.

\begin{rem}
The preceding example shows that, when the overlap marginal is maximally mixed,
noncommutation of the two prescribed marginals prevents Markov feasibility. This
does not mean that Markov states always have commuting overlapping marginals.
A simple example, already implicit in the discussion after Theorem~5 of
\citet{ruskai_2002}, is obtained as follows. Let
$\sigma_1\in\mathbb S_1^+(\hilb_1)$ and
$\sigma_{23}\in\mathbb S_1^+(\hilb_2\otimes\hilb_3)$, and set
$\sigma_2=\trace_3(\sigma_{23})$. Define
$$
\rho=\sigma_1\otimes\sigma_{23}.
$$
Then
$$
\rho_{12}=\sigma_1\otimes\sigma_2,
\qquad
\rho_{23}=\sigma_{23},
\qquad
\rho_2=\sigma_2,
$$
and ${\rm I}(1:3\cd2)_\rho=0$, since system $1$ is independent of systems $2,3$.
Moreover,
$$
\rho_{12}^{1/2}\rho_2^{-1/2}\rho_{23}^{1/2}
=
\sigma_1^{1/2}\otimes\sigma_{23}^{1/2},
$$
and hence the Petz reconstruction gives
$$
\rho_{12}^{1/2}\rho_2^{-1/2}\rho_{23}\rho_2^{-1/2}\rho_{12}^{1/2}
=
\sigma_1\otimes\sigma_{23}
=
\rho.
$$
However, $\rho_{12}$ and $\rho_{23}$ need not commute as operators on
$\hilb_1\otimes\hilb_2\otimes\hilb_3$, since
$$
[\rho_{12},\rho_{23}]
=
\sigma_1\otimes[\sigma_2\otimes I_3,\sigma_{23}],
$$
which may be nonzero.
\end{rem}

\subsection{A Pauli example on a chordal graph}

We now illustrate the chordal trace criterion on the graph in
Figure~\ref{fig:chordal}. The cliques are
$\{1,2,3\}$, $\{2,4,5\}$, $\{2,5,6\}$, and $\{2,6,7\}$, and the separators are $\{2\}$, $\{2,5\}$, and $\{2,6\}$.

We use the following notation. If $P\in\{X,Y,Z\}$, then $P_i$ denotes $P$
acting on the $i$th tensor factor and the identity on all other tensor factors
under consideration. Thus, for example, on the full seven-qubit system,
$$
X_1X_2X_3
=
X\otimes X\otimes X\otimes I\otimes I\otimes I\otimes I,
$$
whereas on the subsystem $\{1,2,3\}$ it denotes $X\otimes X\otimes X$.
Consider now the Pauli words
$$
W_1=X_1X_2X_3,\qquad
W_2=Y_2X_4X_5,\qquad
W_3=Y_2Y_5X_6,\qquad
W_4=Z_2X_6X_7.
$$
They are pairwise anticommuting and satisfy $W_j^2=I$. For parameters
$|a_j|<1$, prescribe the clique marginals
$$
\rho_{123}=\frac18(I+a_1W_1),\qquad
\rho_{245}=\frac18(I+a_2W_2),
$$
$$
\rho_{256}=\frac18(I+a_3W_3),\qquad
\rho_{267}=\frac18(I+a_4W_4),
$$
where each $W_j$ is interpreted on the corresponding three-qubit clique.

All separator marginals are maximally mixed, so the family is locally
consistent. Moreover, if $\sum_{j=1}^4 a_j^2<1$, then it is feasible by the
strictly positive global state
$$
\sigma
=
\frac1{2^7}\left(I+\sum_{j=1}^4 a_jW_j\right),
$$
where now each $W_j$ is interpreted as an operator on the full seven-qubit
system. This follows directly from Lemma~\ref{lem:Pauli-ball}.

We now compute the logarithmic chordal candidate. For this graph,
$$
T(\mathcal R)
=
\exp\{
\log\rho_{123}
+\log\rho_{245}
+\log\rho_{256}
+\log\rho_{267}
-\log\rho_2
-\log\rho_{25}
-\log\rho_{26}
\},
$$
where each marginal is embedded into the full seven-qubit system by tensoring
with identities. Since all separator marginals are maximally mixed,
$\rho_2=\tfrac12I$, $\rho_{25}=\tfrac14 I$, and $\rho_{26}=\tfrac14 I$. By the
logarithmic formula in Lemma~\ref{lem:Pauli-spectral}, applied with one Pauli
word at a time,
$$
\log\{\tfrac18 (I+aW)\}
=
\log(\tfrac{\sqrt{1-a^2}}{8}) I
+
\operatorname{arctanh}(a) W .
$$
Set $b_j=\operatorname{arctanh}(a_j)$ and
$r=(b_1^2+\cdots+b_4^2)^{1/2}$. Substitution gives
$$
T(\mathcal R)
=
\frac{1}{2^7}
\left\{\prod_{j=1}^4(1-a_j^2)^{1/2}\right\}
\exp\left(\sum_{j=1}^4 b_jW_j\right).
$$
Here the scalar power of $2$ is $2^{-12}2^5=2^{-7}$: the four clique
marginals contribute $2^{-12}$, while subtracting the three separator terms
contributes $2^{1+2+2}=2^5$.

The remaining calculation is the same as in
Lemma~\ref{lem:basic-qubit-trace-defect}, with four pairwise anticommuting
Pauli words instead of two. Since
$$
\left(\sum_{j=1}^4 b_jW_j\right)^2=r^2I,
$$
we have
$$
\exp\left(\sum_{j=1}^4 b_jW_j\right)
=
\cosh(r)I
+
\frac{\sinh(r)}{r}\sum_{j=1}^4 b_jW_j,
$$
with the usual interpretation when $r=0$. Since each $W_j$ is traceless,
$$
\trace(T(\mathcal R))
=
\left\{\prod_{j=1}^4(1-a_j^2)^{1/2}\right\}\cosh(r)
=
\frac{\cosh(r)}{\prod_{j=1}^4\cosh(b_j)}.
$$
As in the proof of Lemma~\ref{lem:basic-qubit-trace-defect}, this quantity is
equal to one if at most one of the parameters $a_j$ is nonzero, and is
strictly less than one if at least two of them are nonzero. Hence, by
Theorem~\ref{thm:chordal-trace-criterion}, this locally consistent family has
no quantum Markov completion whenever at least two parameters are nonzero, even
though it has a strictly positive global completion whenever
$\sum_{j=1}^4 a_j^2<1$.

\color{black}

\section*{Acknowledgements}

The authors have benefited from comments of Matthias Christandl and Thomas C. Fraser.
PZ acknowledges the support of Ayudas Fundación BBVA a Proyectos de Investigación Científica 2021, the Spanish Ministry of Economy
and Competitiveness grant PID2022-138268NB-I00, financed by MCIN/AEI/10.13039/501100011033,
FSE+MTM2015-67304-P, FEDER, EU, and Canadian NSERC grant RGPIN-2023-03481. 
PZ is also affiliated with the Serra H\'{u}nter S\`{e}nior Program, and Barcelona School of Economics.

\appendix

\small

\section{A weighted proof of Petz recovery for the partial trace}
\label{app:petz-partial-trace}

This appendix gives a finite-dimensional proof of the equality case in
monotonicity of relative entropy for the partial trace. The result is standard
and goes back to \citet{petz_1986}; see also
\citet{petz2003monotonicity}. We include the proof because it gives exactly the
form of the Petz recovery formula used in the main text, and because the same
argument gives a direct proof of the intersection property for  quantum conditional independence in the strictly
positive case.

The proof is a Hilbert-space version of the relative modular operator
argument. We work with operators as vectors, but use a weighted inner product
adapted to the reference state. In this representation, the partial trace
corresponds to an isometric inclusion of operator spaces, and monotonicity
becomes a comparison between the maximum of a quadratic form over a Hilbert
space and the maximum over a closed subspace.

\subsection{The weighted operator space}

Let $\rho\in\mathbb S_1^+(\hilb)$. On $\cL(\hilb)$ define the weighted inner
product
$$
\langle X,Y\rangle_\rho:=\trace(X^*Y\rho),
\qquad X,Y\in\cL(\hilb).
$$
We write $\cL(\hilb)_\rho$ when we want to emphasize this Hilbert-space
structure. This notation will be used in its full strength in Section~\ref{app:monotonicity}.  For $\tau\in\mathbb S_1^+(\hilb)$, define
$$
\Delta_{\tau,\rho}:\cL(\hilb)_\rho\to\cL(\hilb)_\rho,
\qquad
\Delta_{\tau,\rho}(X):=\tau X\rho^{-1}.
$$

\begin{lem}\label{lem:app-delta-basic}
The operator $\Delta_{\tau,\rho}$ is strictly positive and self-adjoint on
$\cL(\hilb)_\rho$. Moreover,
$$
\langle X,\Delta_{\tau,\rho}Y\rangle_\rho
=
\trace(X^*\tau Y)
\qquad\text{for all }X,Y\in\cL(\hilb).
$$
\end{lem}

\begin{proof}
We have
$$
\langle X,\Delta_{\tau,\rho}Y\rangle_\rho
=
\trace\{X^*(\tau Y\rho^{-1})\rho\}
=
\trace(X^*\tau Y).
$$
To compute the adjoint, suppose that $\Delta_{\tau,\rho}^*$ is the adjoint
with respect to $\langle\cdot,\cdot\rangle_\rho$. Then, for every $Y$,
$$
\<\Delta^*_{\tau,\rho}X,Y\>_\rho=\<X,\Delta_{\tau,\rho}Y\>=\trace(X^*\tau Y)=\trace((\tau X\rho^{-1})^*Y\rho)=\<\tau X \rho^{-1},Y\>_\rho, 
$$
which is just $\<\Delta_{\tau,\rho}X,Y\>_\rho$.
Thus $\Delta_{\tau,\rho}$ is self-adjoint. Finally,
$$
\langle X,\Delta_{\tau,\rho}X\rangle_\rho
=
\trace(X^*\tau X)
=
\trace\{(\tau^{1/2}X)^*(\tau^{1/2}X)\}\ge 0.
$$
Since $\tau$ is strictly positive, equality holds only when $X=0$, so
$\Delta_{\tau,\rho}$ is strictly positive.
\end{proof}

\begin{lem}\label{lem:app-relative-entropy-weighted}
For $\rho,\tau\in\mathbb S_1^+(\hilb)$,
$$
{\rm D}(\rho\|\tau)
=
-\langle I,(\log\Delta_{\tau,\rho})I\rangle_\rho .
$$
\end{lem}

\begin{proof}
Let $L_\tau(X)=\tau X$ and $R_\rho(X)=X\rho$. Then
$\Delta_{\tau,\rho}=L_\tau R_\rho^{-1}$. Since left and right multiplication
commute, $L_\tau$ and $R_\rho$ are commuting positive operators on the operator
space. Hence they admit a simultaneous functional calculus. In particular,
the logarithm of the product $L_\tau R_\rho^{-1}$ is the sum of the logarithms:
$$
\log\Delta_{\tau,\rho}
=
\log(L_\tau R_\rho^{-1})
=
\log L_\tau-\log R_\rho .
$$
Thus, for every operator $X$,
$$
(\log\Delta_{\tau,\rho})(X)
=
(\log\tau)X-X\log\rho .
$$
Taking $X=I$ gives
$
-(\log\Delta_{\tau,\rho})(I)=\log\rho-\log\tau .
$
Therefore, using $\langle X,Y\rangle_\rho=\trace(X^\ast Y\rho)$,
$$
-\langle I,(\log\Delta_{\tau,\rho})I\rangle_\rho
=
\trace\{(\log\rho-\log\tau)\rho\}
=
{\rm D}(\rho\|\tau),
$$
as claimed.
\end{proof}

\subsection{A variational formula for relative entropy}

We use the scalar identity
$$
-\log x
=
\int_0^\infty
\left(\frac{1}{x+t}-\frac{1}{1+t}\right)dt,
\qquad x>0.
$$
By functional calculus and Lemma~\ref{lem:app-relative-entropy-weighted},
\begin{equation}\label{eq:app-relative-entropy-integral}
{\rm D}(\rho\|\tau)
=
\int_0^\infty
\left(
\langle I,(tI+\Delta_{\tau,\rho})^{-1}I\rangle_\rho
-
\frac{1}{1+t}
\right)dt.
\end{equation}

The following elementary variational identity is the only optimization fact
used in the proof. It is just a completion of the square.

\begin{lem}\label{lem:app-shifted-inverse-variational}
Let $\mathcal K$ be a finite-dimensional Hilbert space and let $M$ be a
strictly positive self-adjoint operator on $\mathcal K$. Then, for every
$v\in\mathcal K$,
$$
\langle v,M^{-1}v\rangle
=
\max_{x\in\mathcal K}
\left\{
\langle x,v\rangle+\langle v,x\rangle-\langle x,Mx\rangle
\right\}.
$$
The maximum is attained at the unique point $x=M^{-1}v$.
\end{lem}

\begin{proof}
For every $x\in\mathcal K$,
$$
\langle x,v\rangle+\langle v,x\rangle-\langle x,Mx\rangle
=
\langle M^{-1}v,v\rangle
-
\langle x-M^{-1}v,M(x-M^{-1}v)\rangle.
$$
The first term on the right does not depend on $x$. The second term is nonpositive and vanishes exactly when
$x=M^{-1}v$.
\end{proof}

For $t>0$, define the quadratic function
\begin{equation}\label{eq:app-ft-definition}
f_t^{\rho,\tau}(X)
:=
\langle X,I\rangle_\rho+
\langle I,X\rangle_\rho-
\langle X,(tI+\Delta_{\tau,\rho})X\rangle_\rho,
\qquad X\in\cL(\hilb).
\end{equation}
In trace notation,
\begin{equation}\label{eq:app-ft-trace-form}
f_t^{\rho,\tau}(X)
=
\trace(X^*\rho)+\trace(X\rho)
-
t\trace(X^*X\rho)-\trace(X^*\tau X).
\end{equation}

\begin{prop}
\label{prop:app-relative-entropy-max-formula}
For $\rho,\tau\in\mathbb S_1^+(\hilb)$,
$$
{\rm D}(\rho\|\tau)
=
\int_0^\infty
\left(
\max_{X\in\cL(\hilb)} f_t^{\rho,\tau}(X)-\frac{1}{1+t}
\right)dt.
$$
For each $t>0$, the maximum is attained at the unique point
\begin{equation}\label{eq:app-optimizer-definition}
X_t^{\rho,\tau}
=
(tI+\Delta_{\tau,\rho})^{-1}I.
\end{equation}
\end{prop}

\begin{proof}
Apply Lemma~\ref{lem:app-shifted-inverse-variational} with
$\mathcal K=\cL(\hilb)_\rho$, $M=tI+\Delta_{\tau,\rho}$, and $v=I$. This gives
$$
\max_{X\in\cL(\hilb)} f_t^{\rho,\tau}(X)
=
\langle I,(tI+\Delta_{\tau,\rho})^{-1}I\rangle_\rho,
$$
with unique maximizer $X_t^{\rho,\tau}$ in \eqref{eq:app-optimizer-definition}. Substituting this
identity into \eqref{eq:app-relative-entropy-integral} gives the result.
\end{proof}

\subsection{Monotonicity of the partial trace}\label{app:monotonicity}

Let $\rho,\tau\in\mathbb S_1^+(\hilb_A\otimes\hilb_B)$, and write
$\rho_A=\trace_B(\rho)$ and $\tau_A=\trace_B(\tau)$. Define
$$
J_A:\cL(\hilb_A)_{\rho_A}\to\cL(\hilb_A\otimes\hilb_B)_\rho,
\qquad
J_A(X)=X\otimes I_B.
$$
Throughout this appendix, $J_A^*$ denotes the adjoint of $J_A$ with respect to
the weighted inner products on the domain and codomain.

\begin{lem}\label{lem:app-JA-basic}
The map $J_A$ has the following properties:
\begin{enumerate}[(i)]
\item $J_A$ is an isometry:
$$
\langle J_A(X),J_A(Y)\rangle_\rho
=
\langle X,Y\rangle_{\rho_A}
\qquad\text{for all }X,Y\in\cL(\hilb_A).
$$
\item It maps identity to identity: $J_A(I_A)=I_{A\cup B}$.
\item $J_A$ is multiplicative and preserves adjoints:
$$
J_A(XY)=J_A(X)J_A(Y),
\qquad
J_A(X^*)=J_A(X)^* .
$$
\item It connects $\Delta_{\tau_A,\rho_A}$ and $\Delta_{\tau,\rho}$:
$$
J_A^*\Delta_{\tau,\rho} J_A=\Delta_{\tau_A,\rho_A}.
$$
\end{enumerate}
\end{lem}
\begin{proof}
For (i),
$$
\langle J_A(X),J_A(Y)\rangle_\rho
=
\trace\{(X^*Y\otimes I_B)\rho\}
=
\trace(X^*Y\rho_A)
=
\langle X,Y\rangle_{\rho_A}.
$$
Part (ii) is immediate from the definition $J_A(X)=X\otimes I_B$.

For (iii), let $X,Y\in\cL(\hilb_A)$. Then
$$
J_A(XY)=XY\otimes I_B
=
(X\otimes I_B)(Y\otimes I_B)
=
J_A(X)J_A(Y),
$$
and
$$
J_A(X^*)=X^*\otimes I_B
=
(X\otimes I_B)^*
=
J_A(X)^*.
$$

For (iv), by Lemma~\ref{lem:app-delta-basic},
\begin{align*}
\langle J_A(X),\Delta_{\tau,\rho} J_A(Y)\rangle_\rho
&=
\trace\{(X^*\otimes I_B)\tau(Y\otimes I_B)\} \\
&=
\trace(X^*\tau_A Y) \\
&=
\langle X,\Delta_{\tau_A,\rho_A}Y\rangle_{\rho_A}.
\end{align*}
Since this holds for all $X,Y\in\cL(\hilb_A)$, we get
$
J_A^*\Delta_{\tau,\rho}J_A=\Delta_{\tau_A,\rho_A}.
$
\end{proof}

We conclude the following result.

\begin{lem}\label{lem:app-restriction-weighted}
For every $t>0$ and every $X\in\cL(\hilb_A)$,
$$
f_t^{\rho,\tau}(J_A X)=f_t^{\rho_A,\tau_A}(X).
$$
\end{lem}
\begin{proof}
Recall that
$$
f_t^{\rho,\tau}(Z)
=
\langle Z,I_{A\cup B}\rangle_\rho
+
\langle I_{A\cup B},Z\rangle_\rho
-
t\langle Z,Z\rangle_\rho
-
\langle Z,\Delta_{\tau,\rho} Z\rangle_\rho .
$$
We apply this with $Z=J_A X$. By parts (i) and (ii) of
Lemma~\ref{lem:app-JA-basic},
$$
\langle J_A X,I_{A\cup B}\rangle_\rho
=
\langle X,I_A\rangle_{\rho_A},
\qquad
\langle I_{A\cup B},J_A X\rangle_\rho
=
\langle I_A,X\rangle_{\rho_A},
$$
and
$$
\langle J_A X,J_A X\rangle_\rho
=
\langle X,X\rangle_{\rho_A}.
$$
By part (iv) of Lemma~\ref{lem:app-JA-basic},
$$
\langle J_A X,\Delta_{\tau,\rho}J_A X\rangle_\rho
=
\langle X,J_A^*\Delta_{\tau,\rho}J_A X\rangle_{\rho_A}
=
\langle X,\Delta_{\tau_A,\rho_A}X\rangle_{\rho_A}.
$$
Substituting these identities into the definition of $f_t^{\rho,\tau}$ gives
$f_t^{\rho,\tau}(J_A X)=f_t^{\rho_A,\tau_A}(X)$.
\end{proof}

\begin{prop}[Monotonicity for the partial trace]
\label{prop:app-monotonicity-partial-trace}
Let $\rho,\tau\in\mathbb S_1^+(\hilb_A\otimes\hilb_B)$. Then
$$
{\rm D}(\rho\|\tau)\ge {\rm D}(\rho_A\|\tau_A).
$$
\end{prop}

\begin{proof}
By Proposition~\ref{prop:app-relative-entropy-max-formula},
\begin{equation}\label{eq:diffD}
{\rm D}(\rho\|\tau)-{\rm D}(\rho_A\|\tau_A)
=
\int_0^\infty
\left(
\max_{Z\in\cL(\hilb_A\otimes\hilb_B)}f_t^{\rho,\tau}(Z)
-
\max_{X\in\cL(\hilb_A)}f_t^{\rho_A,\tau_A}(X)
\right)dt .
\end{equation}
By Lemma~\ref{lem:app-restriction-weighted},
$$
\max_{X\in\cL(\hilb_A)}f_t^{\rho_A,\tau_A}(X)
=
\max_{X\in\cL(\hilb_A)}f_t^{\rho,\tau}(J_A X).
$$
The right-hand side is the maximum of $f_t^{\rho,\tau}$ over the subspace
$\im(J_A)$. This is no larger than the maximum over the full space
$\cL(\hilb_A\otimes\hilb_B)$. Integrating over $t>0$ proves the claim.
\end{proof}

\begin{rem}
Proposition~\ref{prop:app-monotonicity-partial-trace} is the special case of
the data-processing inequality for the trace-preserving completely positive map
$\trace_B$. The general theorem says that
${\rm D}(\Phi(\rho)\|\Phi(\tau))\le {\rm D}(\rho\|\tau)$ for every trace-preserving
completely positive map $\Phi$; see \citet{lindblad1975completely}. The point of Proposition~\ref{prop:app-monotonicity-partial-trace} is not to reprove the full theorem, but to isolate the
finite-dimensional geometry of the partial trace. In this form, equality can be
read off from the variational maximizers.
\end{rem}

\subsection{Equality in monotonicity and Petz recovery}

The previous proof compares two variational problems: the global maximization
over $\cL(\hilb_A\otimes\hilb_B)$ and the restricted maximization over
$\im(J_A)$. Since the global quadratic problem has a unique maximizer, equality
can hold only if the global maximizer already belongs to the subspace
$\im(J_A)$.

\begin{prop}[Equality criterion for the partial trace]
\label{prop:app-equality-criterion}
Let $\rho,\tau\in\mathbb S_1^+(\hilb_A\otimes\hilb_B)$, and let $X_t^{\rho,\tau}$ be as in \eqref{eq:app-optimizer-definition}. Then
$$
{\rm D}(\rho\|\tau)={\rm D}(\rho_A\|\tau_A)
$$
if and only if, for every $t>0$ there exists $Y_t\in\cL(\hilb_A)$ such that
$
X_t^{\rho,\tau}=J_A(Y_t)=Y_t\otimes I_B.
$
In this case,
$$
Y_t=(tI+\Delta_{\tau_A,\rho_A})^{-1}I_A.
$$
\end{prop}
\begin{proof}
Using \eqref{eq:diffD} and Lemma~\ref{lem:app-restriction-weighted}, we can write
$$
{\rm D}(\rho\|\tau)-{\rm D}(\rho_A\|\tau_A)
=
\int_0^\infty
\Bigl(
\max_{Z\in\cL(\hilb_A\otimes\hilb_B)}f_t^{\rho,\tau}(Z) -
\max_{X\in\cL(\hilb_A)}f_t^{\rho,\tau}(J_A X)
\Bigr)dt .
$$
The second maximum is the maximum of $f_t^{\rho,\tau}$ over the subspace
$\im(J_A)$. Hence the integrand is nonnegative for every $t>0$. Since it is also
continuous in $t$, the integral vanishes if and only if the integrand vanishes
for every $t>0$.

Fix $t>0$. By Proposition~\ref{prop:app-relative-entropy-max-formula}, the full maximum of
$f_t^{\rho,\tau}$ is attained at the unique point $X_t^{\rho,\tau}$. Therefore the full maximum and the restricted maximum over
$\im(J_A)$ agree if and only if $X_t^{\rho,\tau}\in\im(J_A)$. Equivalently, there
exists $Y_t\in\cL(\hilb_A)$ such that
$$
X_t^{\rho,\tau}=J_A(Y_t)=Y_t\otimes I_B .
$$
This proves the stated equivalence.

It remains only to identify $Y_t$. If $X_t^{\rho,\tau}=J_A(Y_t)$, then
Lemma~\ref{lem:app-restriction-weighted} implies that $Y_t$ maximizes
$f_t^{\rho_A,\tau_A}$. By uniqueness of the marginal maximizer,
$
Y_t=(tI+\Delta_{\tau_A,\rho_A})^{-1}I_A .
$
\end{proof}
We now convert the equality criterion into the Petz recovery formula.

\begin{prop}[Petz recovery for the partial trace]
\label{prop:app-petz-partial-trace}
Let $\rho,\tau\in\mathbb S_1^+(\hilb_A\otimes\hilb_B)$. Then
$
{\rm D}(\rho\|\tau)={\rm D}(\rho_A\|\tau_A)
$
if and only if
\begin{equation}\label{eq:app-petz-recovery}
\rho=
\tau^{1/2}
\bigl(
\tau_A^{-1/2}\rho_A\tau_A^{-1/2}\otimes I_B
\bigr)
\tau^{1/2}.
\end{equation}
\end{prop}

\begin{proof}
Assume first that equality holds in monotonicity. By
Proposition~\ref{prop:app-equality-criterion}, for every $t>0$,
\begin{equation}\label{eq:app-inverse-intertwining}
(tI+\Delta_{\tau,\rho})^{-1}I_{A\cup B}
=
J_A\bigl((tI+\Delta_{\tau_A,\rho_A})^{-1}I_A\bigr).
\end{equation}
From this, we first derive the same identity with the function $s^{-1/2}$ in
place of $s\mapsto (t+s)^{-1}$. For all sufficiently large $t$,
$$
(tI+\Delta)^{-1}
=
t^{-1}\sum_{k=0}^\infty (-1)^k t^{-k}\Delta^k .
$$
Applying this expansion to both sides of
\eqref{eq:app-inverse-intertwining}, we obtain two vector-valued power series in
$t^{-1}$ that agree for all sufficiently large $t$. Comparing coefficients gives
\begin{equation}\label{eq:app-power-intertwining}
\Delta_{\tau,\rho}^k I_{A\cup B}
=
J_A\bigl(\Delta_{\tau_A,\rho_A}^k I_A\bigr),
\qquad k\ge 0.
\end{equation}
Hence the same identity holds with $p(\Delta)$ in place of $\Delta^k$, for every
polynomial $p$. Since the spaces are finite dimensional and the spectra of
$\Delta_{\tau,\rho}$ and $\Delta_{\tau_A,\rho_A}$ are contained in $(0,\infty)$,
we can choose a polynomial $p$ such that $p(s)=s^{-1/2}$ on
$
{\rm spec}(\Delta_{\tau,\rho})\cup{\rm spec}(\Delta_{\tau_A,\rho_A}).
$
By functional calculus,
$
p(\Delta_{\tau,\rho})=\Delta_{\tau,\rho}^{-1/2}
$
and
$
p(\Delta_{\tau_A,\rho_A})=\Delta_{\tau_A,\rho_A}^{-1/2}.
$
Thus \eqref{eq:app-power-intertwining} gives
\begin{equation}\label{eq:app-minus-half-intertwining}
\Delta_{\tau,\rho}^{-1/2}I_{A\cup B}
=
J_A\bigl(\Delta_{\tau_A,\rho_A}^{-1/2}I_A\bigr).
\end{equation}
Now
$$
\Delta_{\tau,\rho}^{-1/2}
=
L_{\tau^{-1/2}}R_{\rho^{1/2}},
\qquad
\Delta_{\tau_A,\rho_A}^{-1/2}
=
L_{\tau_A^{-1/2}}R_{\rho_A^{1/2}}.
$$
Applying both sides to the identity operators, \eqref{eq:app-minus-half-intertwining}
becomes
$$
\tau^{-1/2}\rho^{1/2}
=
J_A(\tau_A^{-1/2}\rho_A^{1/2}).
$$
Multiplying on the left by $\tau^{1/2}$ gives
\begin{equation}\label{eq:app-amplitude-identity}
\rho^{1/2}
=
\tau^{1/2}M,
\qquad
M:=J_A(\tau_A^{-1/2}\rho_A^{1/2}).
\end{equation}
Therefore
$$
\rho
=
\rho^{1/2}(\rho^{1/2})^*
=
\tau^{1/2}MM^*\tau^{1/2}.
$$
It remains to compute $MM^*$. By part (iii) of
Lemma~\ref{lem:app-JA-basic},
\begin{align*}
MM^*
&=
J_A(\tau_A^{-1/2}\rho_A^{1/2})
J_A(\tau_A^{-1/2}\rho_A^{1/2})^* \\
&=
J_A(\tau_A^{-1/2}\rho_A^{1/2})
J_A(\rho_A^{1/2}\tau_A^{-1/2}) \\
&=
J_A(\tau_A^{-1/2}\rho_A\tau_A^{-1/2}).
\end{align*}
Thus
$
\rho
=
\tau^{1/2}
J_A(\tau_A^{-1/2}\rho_A\tau_A^{-1/2})
\tau^{1/2},
$
which is \eqref{eq:app-petz-recovery}. 

Conversely, assume \eqref{eq:app-petz-recovery}. Define
$$
\mathcal R_\tau(W)
:=
\tau^{1/2}
J_A(\tau_A^{-1/2}W\tau_A^{-1/2})
\tau^{1/2},
\qquad W\in\cL(\hilb_A).
$$
This map is completely positive. It is trace preserving because
\begin{align*}
\trace\{\mathcal R_\tau(W)\}
&=
\trace\{J_A(\tau_A^{-1/2}W\tau_A^{-1/2})\tau\} =
\trace\{(\tau_A^{-1/2}W\tau_A^{-1/2})\tau_A\} =
\trace(W).
\end{align*}
Moreover,
$
\mathcal R_\tau(\tau_A)=\tau$ and
$\mathcal R_\tau(\rho_A)=\rho$.
By monotonicity under the partial trace,
$
{\rm D}(\rho\|\tau)\ge {\rm D}(\rho_A\|\tau_A).
$
By the general data-processing inequality for trace-preserving completely
positive maps, applied to $\mathcal R_\tau$,
$$
{\rm D}(\rho_A\|\tau_A)
\ge
{\rm D}(\mathcal R_\tau(\rho_A)\|\mathcal R_\tau(\tau_A))
=
{\rm D}(\rho\|\tau).
$$
Thus equality holds.
\end{proof}

\subsection{Conditional independence as Petz recovery}
\label{sec:app-conditional-independence-petz}

We next specialize Proposition~\ref{prop:app-petz-partial-trace} to conditional
mutual information. The equality case of strong subadditivity, or equivalently
the structure of states with vanishing conditional mutual information, was
characterized by \citet{hayden_jozsa_petz_winter_2004}. The argument below does
not use that decomposition; it uses only the partial-trace recovery formula
proved above.

Let $\sigma\in\mathbb S_1^+(\hilb_A\otimes\hilb_B\otimes\hilb_C)$ and set
$
\tau=\sigma_A\otimes\sigma_{B\cup C}.
$
Then
$
\tau_{A\cup C}=\sigma_A\otimes\sigma_C.
$
Moreover,
\begin{align*}
{\rm D}(\sigma\|\tau)
&=
\trace(\sigma\log\sigma)
-
\trace\{\sigma(\log\sigma_A+\log\sigma_{B\cup C})\} =
-{\rm S}(\sigma)+{\rm S}(\sigma_A)+{\rm S}(\sigma_{B\cup C}),
\end{align*}
and
$
{\rm D}(\sigma_{A\cup C}\|\tau_{A\cup C})
=
-{\rm S}(\sigma_{A\cup C})+{\rm S}(\sigma_A)+{\rm S}(\sigma_C).
$
Thus
\begin{equation}\label{eq:app-cmi-relative-entropy-difference}
{\rm D}(\sigma\|\tau)
-
{\rm D}(\sigma_{A\cup C}\|\tau_{A\cup C})
=
{\rm I}(A:B\cd C)_\sigma.
\end{equation}
Therefore ${\rm I}(A:B\cd C)_\sigma=0$ is exactly equality in monotonicity for the
partial trace over $B$, applied to the pair $(\sigma,\tau)$.

\begin{prop}\label{prop:app-conditional-independence-petz}
Let $\sigma\in\mathbb S_1^+(\hilb_A\otimes\hilb_B\otimes\hilb_C)$. Then
${\rm I}(A:B\cd C)_\sigma=0$ if and only if
\begin{equation}\label{eq:app-conditional-independence-petz}
\sigma
=
\sigma_{B\cup C}^{1/2}\sigma_C^{-1/2}
\sigma_{A\cup C}
\sigma_C^{-1/2}\sigma_{B\cup C}^{1/2},
\end{equation}
where operators on $B\cup C$ act trivially on $A$, and operators on
$A\cup C$ act trivially on $B$.
\end{prop}

\begin{proof}
By \eqref{eq:app-cmi-relative-entropy-difference},
${\rm I}(A:B\cd C)_\sigma=0$ is equivalent to
$
{\rm D}(\sigma\|\tau)={\rm D}(\sigma_{A\cup C}\|\tau_{A\cup C})
$,
where $\tau=\sigma_A\otimes\sigma_{B\cup C}$.
Applying Proposition~\ref{prop:app-petz-partial-trace} with the retained system
$A\cup C$ and the traced system $B$ gives
$$
\sigma
=
\tau^{1/2}
\left(
\tau_{A\cup C}^{-1/2}
\sigma_{A\cup C}
\tau_{A\cup C}^{-1/2}
\otimes I_B
\right)
\tau^{1/2}.
$$
Since
$
\tau^{1/2}=\sigma_A^{1/2}\otimes\sigma_{B\cup C}^{1/2}
$
and
$
\tau_{A\cup C}^{-1/2}=\sigma_A^{-1/2}\otimes\sigma_C^{-1/2},
$
the factors $\sigma_A^{1/2}$ and $\sigma_A^{-1/2}$ cancel on the left and on
the right. With the convention that all operators are embedded into
$\cL(\hilb_A\otimes\hilb_B\otimes\hilb_C)$, this gives
\eqref{eq:app-conditional-independence-petz}.
\end{proof}

\subsection{Proof of Proposition~\ref{prop:intersection}}
\label{sec:proof-inters}

We prove the intersection property by applying the equality criterion for the
partial trace, Proposition~\ref{prop:app-equality-criterion}. Let
$\rho\in\mathbb S_1^+(\hilb_{A\cup B\cup C\cup D})$ be strictly positive and
assume
$$
A\ciph B\cd(C\cup D)\,[\rho]
\qquad\text{and}\qquad
A\ciph D\cd(B\cup C)\,[\rho].
$$
Set
$
\tau=\rho_A\otimes\rho_{B\cup C\cup D}.
$
Then
$$
\tau_{A\cup C\cup D}=\rho_A\otimes\rho_{C\cup D},
\qquad
\tau_{A\cup B\cup C}=\rho_A\otimes\rho_{B\cup C},
\qquad
\tau_{A\cup C}=\rho_A\otimes\rho_C .
$$
We first relate the two assumed conditional independences to equality cases in
monotonicity. Since $\log\tau=\log\rho_A\otimes I_{B\cup C\cup D}
+I_A\otimes\log\rho_{B\cup C\cup D}$, we have
$$
{\rm D}(\rho\|\tau)
=
-{\rm S}(\rho)+{\rm S}(\rho_A)+{\rm S}(\rho_{B\cup C\cup D}).
$$
Similarly,
$$
{\rm D}(\rho_{A\cup C\cup D}\|\tau_{A\cup C\cup D})
=
-{\rm S}(\rho_{A\cup C\cup D})+{\rm S}(\rho_A)+{\rm S}(\rho_{C\cup D}).
$$
Therefore
$$
{\rm D}(\rho\|\tau)
-
{\rm D}(\rho_{A\cup C\cup D}\|\tau_{A\cup C\cup D})
=
{\rm S}(\rho_{A\cup C\cup D})+{\rm S}(\rho_{B\cup C\cup D})  
-{\rm S}(\rho_{C\cup D})-{\rm S}(\rho),
$$
which is equal to ${\rm I}(A:B\cd C\cup D)_\rho$.
Hence the assumption $A\ciph B\cd(C\cup D)\,[\rho]$ is exactly equality in
Proposition~\ref{prop:app-monotonicity-partial-trace} for the partial trace
over $B$, applied to the pair $(\rho,\tau)$.

The same computation after tracing out $D$ gives
$$
{\rm D}(\rho\|\tau)
-
{\rm D}(\rho_{A\cup B\cup C}\|\tau_{A\cup B\cup C})
=
{\rm I}(A:D\cd B\cup C)_\rho .
$$
Thus the assumption $A\ciph D\cd(B\cup C)\,[\rho]$ is exactly equality in
Proposition~\ref{prop:app-monotonicity-partial-trace} for the partial trace
over $D$, again applied to the pair $(\rho,\tau)$. For $t>0$, let
$$
X_t:=X_t^{\rho,\tau}
=
(tI+\Delta_{\tau,\rho})^{-1}I_{A\cup B\cup C\cup D},
$$
as in \eqref{eq:app-optimizer-definition}. By
Proposition~\ref{prop:app-equality-criterion}, equality for the partial trace
over $B$ implies that $X_t$ is constant in the $B$-coordinate:
$$
X_t\in
\cL(\hilb_A\otimes\hilb_C\otimes\hilb_D)\otimes I_B .
$$
Similarly, equality for the partial trace over $D$ implies that $X_t$ is
constant in the $D$-coordinate:
$$
X_t\in
\cL(\hilb_A\otimes\hilb_B\otimes\hilb_C)\otimes I_D .
$$
Equivalently, after identifying all spaces inside
$\cL(\hilb_A\otimes\hilb_B\otimes\hilb_C\otimes\hilb_D)$,
\begin{align*}
X_t
&\in
\bigl(\cL(\hilb_A\otimes\hilb_C\otimes\hilb_D)\otimes I_B\bigr)
\cap
\bigl(\cL(\hilb_A\otimes\hilb_B\otimes\hilb_C)\otimes I_D\bigr).
\end{align*}
This intersection is
$
\cL(\hilb_A\otimes\hilb_C)\otimes I_B\otimes I_D .
$
To see this, choose bases of $\cL(\hilb_B)$ and $\cL(\hilb_D)$ whose first
elements are $I_B$ and $I_D$. Expanding an operator in the corresponding
tensor-product basis, membership in the first subspace forces all non-identity
basis coefficients in the $B$-coordinate to vanish, while membership in the
second subspace forces all non-identity basis coefficients in the $D$-coordinate
to vanish. Hence membership in both subspaces forces identity components in
both the $B$- and $D$-coordinates. Thus, for every $t>0$,
$$
X_t\in
\cL(\hilb_A\otimes\hilb_C)\otimes I_B\otimes I_D .
$$
Applying Proposition~\ref{prop:app-equality-criterion} once more, now to the
partial trace over $B\cup D$, gives
$$
{\rm D}(\rho\|\tau)={\rm D}(\rho_{A\cup C}\|\tau_{A\cup C}).
$$
Finally, using $\tau_{A\cup C}=\rho_A\otimes\rho_C$, we compute
$$
{\rm D}(\rho\|\tau)-{\rm D}(\rho_{A\cup C}\|\tau_{A\cup C})
=
{\rm I}(A:B\cup D\cd C)_\rho .
$$
Therefore ${\rm I}(A:B\cup D\cd C)_\rho=0$, which means
$
A\ciph (B\cup D)\cd C\,[\rho].
$
This proves the intersection property.

\end{document}